\documentclass[aps,pra,twocolumn,superscriptaddress,longbibliography]{revtex4-2}
\usepackage[T1]{fontenc}
\usepackage[utf8]{inputenc}
\usepackage{amsmath,amssymb,mathtools,mathrsfs,amsfonts}
\usepackage{bm,booktabs,array,graphicx,tabularx}
\usepackage{xcolor}
\definecolor{RevisionBlue}{RGB}{0,75,165}
\newif\ifshowchanges
\showchangesfalse
\DeclareRobustCommand{\rev}[1]{{\ifshowchanges\color{RevisionBlue}\fi#1}}

\usepackage[hidelinks]{hyperref}
\usepackage{microtype}
\allowdisplaybreaks[2]
\newtheorem{theorem}{Theorem}
\newtheorem{lemma}{Lemma}
\newtheorem{corollary}[theorem]{Corollary}

\newtheorem{proposition}{Proposition}
\newtheorem{definition}{Definition}
\newtheorem{remark}{Remark}
\newenvironment{proof}[1][Proof]{\par\noindent\emph{#1.}\ }{\hfill$\square$\par}
\newcommand{\cQ}{\mathcal{Q}}
\newcommand{\cL}{\mathcal{L}}
\newcommand{\cS}{\mathcal{S}}
\newcommand{\cE}{\mathcal{E}}
\newcommand{\vB}{\vec{B}}
\newcommand{\vQ}{\vec{Q}}
\newcommand{\vL}{\vec{L}}

\newcommand{\vtheta}{\vec{\theta}}

\newcommand{\Ext}{\operatorname{Ext}}
\newcommand{\conv}{\operatorname{conv}}
\newcommand{\cone}{\operatorname{cone}}
\newcommand{\tr}{\operatorname{tr}}
\newcommand{\rank}{\operatorname{rank}}
\newcommand{\intt}{\operatorname{int}}
\newcommand{\sgn}{\operatorname{sgn}}
\newcommand{\scrN}{\mathcal{N}}
\newcommand{\Bell}{\widehat{B}}
\newcommand{\eps}{\varepsilon}
\begin{document}
\title{Exact Identification of Quantum Correlations in the Minimal Bell Scenario}
\author{Liang-Liang Sun}
\affiliation{Department of Modern Physics and National Laboratory for Physical Sciences at Microscale, University of Science and Technology of China, Hefei, Anhui 230026, China}
\author{Jie Wang}
\email{wangjie212@amss.ac.cn}
\affiliation{State Key Laboratory of Mathematical Sciences, Academy of Mathematics and Systems Science, Chinese Academy of Sciences, Beijing 100190, China}
\author{Sixia Yu}
\email{yusixia@ustc.edu.cn}
\affiliation{Department of Modern Physics and National Laboratory for Physical Sciences at Microscale, University of Science and Technology of China, Hefei, Anhui 230026, China}
\date{October 17, 2026}
\begin{abstract}
Bell nonlocality places the long-standing philosophical disputes over the nature of quantum theory on a mathematically rigorous, experimentally testable footing, positioning quantum mechanics between local realism and the most general no-signaling theories. Yet the precise quantum boundary has remained elusive  even in the minimal Bell scenario, where the extreme points generating the set are known. Here we characterize this boundary through a computable necessary and sufficient criterion: a correlation is quantum if and only if it can be decomposed into an extreme quantum correlation and a classical correlation. The criterion follows from our proof that every Bell functional exhibiting quantum advantage admits a unique quantum maximizer. We further classify the finitely many maximal-face structures containing each nonlocal exposed point according to their local deterministic vertices. By sharply distinguishing physically realizable correlations from nonphysical ones, these results resolve a long-standing open problem and provide a new foundation for identifying the principles underlying quantum theory and determining the ultimate information-theoretic limits of quantum correlations.
\end{abstract}
\maketitle

Bell's theorem provides a means to distinguish theories of fundamentally different natures by the correlations they permit \cite{Bell1964,Brunner2014}: classical correlations are bounded by Bell inequalities, the outer frontier is set by the no-signaling principle, and quantum correlations occupy a distinctive middle ground. They enable information-processing advantages beyond the classical theory \cite{Ekert1991,Acin2007,Pironio2010}, but they stop short of the full range of correlations compatible with the no-signaling principle \cite{Tsirelson1980,Popescu1994,BarrettPironio2005,Barrett2005}. This very gap poses a fundamental question: does some additional physical principle prevent quantum correlations from attaining the full strength allowed by no-signaling alone \cite{Popescu1994,Brassard2006,Pawlowski2009,NavascuesWunderlich2010,Fritz2013,Navascues2015}? At the heart of these efforts lies the precise characterization of the boundary of quantum correlations---a problem that has remained open even in the minimal Bell scenario.

Determining this boundary amounts to translating the mathematical structure of quantum theory into operational constraints on the observed statistics. The canonical tool is the semidefinite-programming hierarchy---a nested sequence of outer approximations that converges to the quantum set only asymptotically \cite{NPA2007,NPA2008,Doherty2008}, so that membership cannot in general be certified at a fixed finite level. The minimal Bell scenario offers a more tractable arena \cite{Landau1988,Masanes2003,Goh2018}. Tailored analyses have led to a sufficient criterion for correlations with vanishing marginals \cite{Landau1988,Masanes2003,Ishizaka2018}. By a convexity argument, Carath\'{e}odory's theorem \cite{Rockafellar1970} guarantees that any correlation decomposes into at most nine extreme points; although these points admit an analytical characterization \cite{Goh2018,Masanes2005,WangWuScarani2016,WolfeYelin2012,BarizienBancal2025}, the number of free parameters renders such a direct decomposition intractable. Dually, eight supporting hyperplanes \cite{Goh2018,DonohueWolfe2015} associated with the extreme points of the Clauser--Horne--Shimony--Holt (CHSH) inequality \cite{CHSH1969} have been identified through case-by-case sum-of-squares constructions \cite{Tsirelson1993,Wehner2006,BampsPironio2015}. The complete boundary, comprising infinitely many such hyperplanes, has remained analytically out of reach for such one-by-one techniques. Meanwhile, accumulating evidence has revealed striking regularities peculiar to this scenario that remain unexplained: every Bell functional exhibiting a quantum advantage possesses a unique maximizer, which seems to point to a hidden structure governing the quantum boundary that existing methods have yet to capture.

Here we close this gap by establishing a necessary and sufficient criterion for quantum correlations in the minimal Bell scenario, together with a constructive approach to its implementation. Our key insight comes from the geometry of the quantum set. We prove that every Bell functional exhibiting a strict quantum advantage has a unique quantum maximizer. More generally, any Bell functional maximized at a nonlocal exposed point admits no other nonlocal extreme maximizer, even in the absence of a quantum advantage. These results lead to a decomposition of every quantum correlation into a single correlation generated by a pure two-qubit state with projective measurements and a local correlation. For each nonlocal exposed point, we further reduce the classification of all maximal proper faces containing it to finitely many local-vertex configurations and exact feasibility tests. Exploiting this geometric structure, we reformulate the membership criterion as the feasibility of a system of trigonometric equations in five angular parameters, providing a route to its numerical implementation.

\section{Notions}
The simplest nontrivial Bell scenario involves two parties, Alice and Bob, who share a bipartite state $\rho$ and each choose at random between two binary measurements, $\hat A_x$ and $\hat B_y$ with $x,y\in\{0,1\}$, obtaining outcomes $a,b\in\{\pm1\}$ with probability $p(ab|xy)$. Equivalently, their statistics are captured by eight expectation values---the marginals $\langle A_x\rangle$ and $\langle B_y\rangle$ together with the joint correlators $E_{xy}:=\langle A_xB_y\rangle$---which we collect into a vector $\vQ:=(\langle A_x\rangle,\langle B_y\rangle,\langle A_x B_y\rangle)\in\mathbb R^8$. The set of all such vectors realizable in quantum theory is the quantum set $\cQ\subset\mathbb R^8$. Of particular interest is the subset of quantum correlations that can be realized by projective measurements on a pure two-qubit state \cite{Masanes2005,WangWuScarani2016,BarizienBancal2025}, which can be parametrized by five angles $\vtheta:=(a_0,a_1,b_0,b_1,\theta)$ as $\vQ_{\vtheta}$:
\begin{equation}\label{eq:family-main}
\begin{gathered}\cS:=\{\vQ_{\vtheta}=(A_x,B_y,E_{xy})\},\\ A_x=\cos\theta\cos a_x,\quad B_y=\cos\theta\cos b_y,\\ E_{xy}=\cos a_x\cos b_y+\sin\theta\sin a_x\sin b_y.\end{gathered}
\end{equation}
with $a_x,b_y\in[0,2\pi)$ and $\theta\in[0,\pi/2]$. This set contains all extreme points of $\cQ$ \cite{Goh2018,BarizienBancal2025}; however, membership in $\cS$ does not by itself guarantee extremality, as the extreme points form a subset $\cE\subsetneq\cS$. Sixteen of these extreme points are classical; we denote them by $\cE_L=\{\vL_l\}_{l=1}^{16}$. Their convex hull is the local polytope $\cL$ \cite{Fine1982,Froissart1981,Pitowsky1989}, whose nontrivial facets are given by the saturated CHSH inequalities \cite{CHSH1969,Sliwa2003}
\begin{equation*}
S_\eps:=\sum_{x,y=0}^{1}(-1)^{\eps_{xy}}E_{xy}\le2,
\end{equation*}
where $\eps_{xy}\in\{0,1\}$ are sign patterns satisfying $\sum_{x,y}\eps_{xy}\equiv1\pmod 2$.

The remaining, nonclassical extreme points are either exposed or non-exposed, and we denote the two sets by $\cE_{\rm es}$ and $\cE_{\rm ne}$, respectively, according to whether the point is the unique maximizer of some Bell functional \cite{Goh2018}. A Bell functional is specified by a nonzero coefficient vector $\vB=(\alpha_0,\alpha_1,\beta_0,\beta_1,c_{00},c_{01},c_{10},c_{11})\in\mathbb R^8$, acting linearly on a correlation $\vec{Q}$ via $\beta_{\vec{B}}(\vec{Q}):=\vec{B}\cdot \vec{Q}$. The polar set of $\cQ$ is
\begin{equation}\label{eq:polar}
\mathcal B=\{\vB\in\mathbb R^8:\vB\cdot\vQ\le1\ \forall\vQ\in\cQ\}, \nonumber
\end{equation}
where $Q_{\vB}:=\max_{\vQ\in\cQ}\vB\cdot\vQ$ is the quantum bound of $\vB$. We call a Bell functional frustrated if
\begin{equation}\label{eq:frustration-main}
c_{00}c_{01}c_{10}c_{11}<0. 
\end{equation}
A Bell functional defines a supporting hyperplane $H_{\vB}$ whose intersection with $\cQ$ is the exposed face $F_{\vB}$,
\begin{equation*}
H_{\vB}:=\{\vQ\in\mathbb R^8:\vB\cdot\vQ=Q_{\vB}\},\quad F_{\vB}:=H_{\vB}\cap\cQ. \nonumber
\end{equation*}
If $F_{\vB}$ is a singleton, its unique element is an exposed point; otherwise, the Bell functional exposes a proper face. A proper face is called maximal if there is no larger proper face $F_{\vB'}$ with $F_{\vB}\subsetneq F_{\vB'}$. The Bell functionals maximized at a nonlocal point $\vQ_{\vtheta}\in\cQ$ form its normal cone,
\begin{equation}\label{eq:normal-main}
\scrN(\vQ_{\vtheta})=\{\vB\in\mathbb R^8:\vB\cdot(\vQ-\vQ_{\vtheta})\le0\ \forall\vQ\in\cQ\}. \nonumber
\end{equation}
For each nonzero $\vB\in\scrN(\vQ_{\vtheta})$, the maximizing face is $F_{\vB}=\{\vQ\in\cQ:\vB\cdot\vQ=\vB\cdot\vQ_{\vtheta}\}$. The extreme rays of the cone characterize it completely: every other element is a nonnegative combination of them. At the relative boundary of $\scrN(\vQ_{\vtheta})$, the maximizing face may grow and acquire additional extreme points.

Further background on Bell operators, the overlap parametrization, and the convex geometry used below is collected in Appendix~\ref{app:basic}.

\section{Main results}
The quantum correlation set $\cQ\subset\mathbb R^8$ is compact and convex. Carath\'{e}odory's theorem guarantees that each $\vQ\in\cQ$ can be expressed as a convex combination of at most nine extreme points. This dimension-based bound, however, does not reflect the special geometry of $\cQ$. Exploiting this geometry, we obtain a more structured decomposition, valid throughout the quantum set, involving only one member of $\cS$ and a local behavior.

\begin{theorem}[One-point-plus-local decomposition]\label{thm:decomp}
In the $(2,2,2)$ Bell scenario, a behavior is quantum if and only if it decomposes into a single behavior generated by a pure real two-qubit state with real projective measurements and a local behavior. Equivalently,
\begin{equation}\label{eq:decomp-main}
\begin{gathered}\cQ=\{\lambda \vQ_{\vtheta}+(1-\lambda)\vL:
\\0\le\lambda\le1,\ \vQ_{\vtheta}\in\cS,\ \vL\in\cL\}.\end{gathered}
\end{equation}
\end{theorem}
The proof is given in Appendix~\ref{app:decomposition}, building on the uniqueness theorem established in Appendix~\ref{sec:unique}.

To understand the geometric origin of the decomposition, note that every interior quantum behavior is a convex combination of an interior local point and a boundary point lying in a maximal proper face. The extreme points generating these faces therefore govern the possible decompositions. We therefore study these faces through the normal cone $\scrN(\vQ_{\vtheta})$ of a nonlocal exposed point, which is sufficient for the decomposition theorem. Nonlocal extreme points with a vanishing joint probability are non-exposed \cite{BarizienBancal2025} and lie on explicitly known positivity facets \cite{Goh2018,Hardy1993}: each condition $p(ab|xy)=0$ defines a maximal proper face $F_{ab|xy}=\{\vQ\in\cQ:p(ab|xy)=0\}$. Their supporting positivity hyperplanes are therefore identified directly, without computing any normal cone; this does not determine the full quantum structure within those faces or imply locality.

\begin{theorem}[Nontrivial maximal proper faces]\label{thm:maximal-faces}
In the $(2,2,2)$ Bell scenario, every maximal proper face $F$ containing a nonlocal exposed point $\vQ_{\vtheta}$ has no other nonlocal extreme points. Its remaining extreme points are local deterministic behaviors. More precisely,
\begin{equation}\label{eq:maximal-face}
\begin{gathered}
F=\conv\bigl(\{\vQ_{\vtheta}\}\cup V_F\bigr),\\
V_F:=F\cap\Ext(\cL).
\end{gathered}
\end{equation}
\end{theorem}
This follows from the structure of frustrated maximizing faces (Theorem~\ref{thm:frustrated-face}) proved in Appendix~\ref{sec:unique}; see also Theorem~\ref{thm:no-common-support}. 

\paragraph*{Frustrated signs and quantum advantage.}
Frustrated signs are a necessary condition for a quantum advantage. Local output relabelings allow Alice and Bob to flip their outcomes independently at each input. For $u_x,v_y\in\{\pm1\}$, these transformations act as $A'_x=u_xA_x$, $B'_y=v_yB_y$, and $E'_{xy}=u_xv_yE_{xy}$. Transforming the Bell coefficients simultaneously as $\alpha'_x=u_x\alpha_x$, $\beta'_y=v_y\beta_y$, and $c'_{xy}=u_xv_yc_{xy}$ preserves the Bell value and both optimization bounds, $\beta_{\vB'}(\vQ')=\beta_{\vB}(\vQ)$. Thus, if $\vB$ presents a quantum advantage, so does $\vB'$; the two are equivalent for our purposes.

For a Bell functional $\vB$ (possibly after relabeling) and one of its maximizers $\vQ$, one can construct a local behavior $\vL$ with the same marginals and correlators $E_{xy}^{\rm Loc}:=1-|A_x-B_y|\ge E_{xy}$.  Since $E_{xy}^{\rm Loc}-E_{xy}
=4\min\{p(+-|xy),p(-+|xy)\},$
and the local construction preserves the marginals,
we obtain
\begin{equation}\label{eq:local-comparison-main}
\beta_{\vB}(\vL)-\beta_{\vB}(\vQ)
=4\sum_{x,y}c_{xy}
\min\{p(+-|xy),p(-+|xy)\}. \nonumber
\end{equation}
Thus, when all $c_{xy}\ge0$, the local behavior attains a Bell value no smaller than that of
$\vQ$, with strict inequality whenever $c_{xy}>0$ and both corresponding probabilities
are positive for some input pair. Since local output relabelings can make all
nonzero correlator coefficients nonnegative
whenever $c_{00}c_{01}c_{10}c_{11}\ge0$,
a strict quantum advantage necessarily requires
frustrated signs,
$c_{00}c_{01}c_{10}c_{11}<0$.
Moreover, if a maximizing quantum behavior has
strictly positive joint probabilities, every
nonzero Bell functional supporting it must be
frustrated, even when the quantum and local
bounds coincide.

We prove that every frustrated Bell functional maximized at exposed point 
has at most one nonlocal extreme maximizer,
although local maximizers may coexist.
Consequently, any maximizing face containing
a nonlocal extreme point $\vQ_{\vtheta}$ has
the form $F_{\vB}
=\conv\bigl(\{\vQ_{\vtheta}\}
\cup(F_{\vB}\cap\cL)\bigr).$ In particular, every Bell functional exhibiting
a strict quantum advantage has a unique quantum
maximizer, resolving the conjecture of
Ref.~\cite{Goh2018}.
The proof is given in Appendix~\ref{app:faces}.

This geometric structure implies that every
point on such an exposed face is a convex mixture
of $\vQ_{\vtheta}$ and a local behavior.
The same representation is not restricted to
these faces. Indeed, the minimum-weight
construction in Appendix~\ref{app:decomposition}
establishes the decomposition for interior
quantum behaviors, while compactness extends it
to the boundary of $\cQ$, including points on
positivity faces. Hence every quantum behavior
admits the one-point-plus-local decomposition
stated in Theorem~\ref{thm:decomp}.

Finally, we note that the decomposition theorem also yields a finite-dimensional realization of every quantum behavior. The component $\vQ_{\vtheta}$ admits a two-qubit realization with state $\rho_{\vtheta}$, while any local behavior $\vL$ can be realized by a separable state $\rho_L$ on \rev{$\mathbb C^4\otimes\mathbb C^4$ \cite{DonohueWolfe2015}. Indeed, each local four-dimensional basis labels one of the four predetermined output pairs; a diagonal separable state encoding their joint distribution, with diagonal projective observables, reproduces all local behaviors.} The decomposition admits a direct-sum realization with $\mathcal H_A=\mathcal H_B=\mathbb C^2\oplus\mathbb C^4$, block-diagonal local observables, and $\rho=\lambda\rho_{\vtheta}\oplus(1-\lambda)\rho_L$ supported on the two matched tensor-product sectors. Consequently, a $6\otimes6$ Hilbert space suffices to realize every quantum behavior in the $(2,2,2)$ Bell scenario.

\section{Normal cone of a quantum extreme point}
Our uniqueness theorem makes the face structure finite at each fixed nonlocal exposed point, without requiring its normal cone to be polyhedral.

\begin{corollary}[Finite maximal-face classification]\label{thm:finite-cone}
Fix a nonlocal exposed point $\vQ_{\vtheta}$, and let $\{\vL_i\}_{i=1}^{8}=\Ext(\cL)$ are the eight  local extreme correlations saturating the CHSH inequality of  $\vec{Q}_{\vec{\theta}}$. Every exposed face $F_{\vB}$ containing $\vQ_{\vtheta}$ satisfies
\begin{equation}\label{eq:finite-cone}
F_{\vB}=\conv\bigl(\{\vQ_{\vtheta}\}\cup\{\vL_i:\vB\cdot\vL_i=\vB\cdot \vQ_{\vtheta}\}\bigr).
\end{equation}
Thus only finitely many exposed faces, and in particular finitely many maximal proper faces, contain $\vQ_{\vtheta}$. All such maximal faces are obtained by retaining the inclusion-maximal feasible sets $I\subseteq\{1,\ldots,8\}$ satisfying
\begin{equation*}
\begin{gathered}
\vB\cdot \vQ_{\vtheta}=\max_{\vtheta'}\vB\cdot\vQ_{\vtheta'}=1,\\
\vB\cdot\vL_i=1\ (i\in I),\qquad
\vB\cdot\vL_j<1\ (j\notin I)
\end{gathered}
\end{equation*}
for some Bell functional $\vB$.
\end{corollary}

Indeed, by the maximization requirement, $\vB \in \scrN(\vQ_{\vtheta})$ if and only if $\vB \cdot \vQ_{\vtheta} = Q_{\vB}$, which implies $\vB \cdot \partial \vQ_{\vtheta} / \partial \theta_k = 0$ for all $\theta_{k}\in \{a_{x}, b_{y}, \theta\}_{x, y=0, 1}$. With $[T_{\vtheta}]_{ik} := \partial (\vQ_{\vtheta})_i / \partial \theta_k$ the $8 \times 5$ tangent matrix, these conditions read $T_{\vtheta}^{T} \vB = 0$. For $\theta \ne \pi/2$ the tangent matrix has rank five, imposing five independent linear constraints on $\vB$. At $\theta = \pi/2$ the tangent rank drops to four at regular points; however, a common shift $a_x \mapsto a_x + t$, $b_y \mapsto b_y + t$ leaves the behavior unchanged, which supplies the additional constraint $\sum_{x=0}^{1} \alpha_x \sin a_x + \sum_{y=0}^{1} \beta_y \sin b_y = 0$. Fixing the overall scale of $\vB$ by $\vB \cdot \vQ_{\vtheta} = 1$, these six independent linear constraints in total reduce the candidate normals from $\mathbb{R}^8$ to a two-parameter affine family $\vB(c_1, c_2)$. For a supporting face through a nonlocal exposed point, all additional extreme points must be local deterministic vertices. Imposing the saturation conditions $\vB(c_1, c_2) \cdot \vL_i = 1$ at two such local vertices typically fixes the remaining parameters, provided that these constraints are independent. We denote the resulting candidate by $\vB_{ij}$ and retain it only if $\vB_{ij} \cdot \vL_k \le 1$ for all other local vertices $\vL_k$; the index set of the local vertices entering the decomposition is denoted by $I_{\vtheta}$. As an example, consider the extreme point $\vQ_{\vtheta}$ with parameters $(a_0, a_1, b_0, b_1, \theta) = (5\pi/6,\, 3\pi/2,\, 2\pi/3,\, 7\pi/5,\, 11\pi/20)$; its normal cone---a pentagon generated by five such functionals $\vB_{ij}$---is computed in the appendix.

A certified face description can also restrict the classical part of a correlation in its decompositions. Since the origin lies in the interior of $\cQ$, every $\vQ\in\cQ$ lies on a segment joining the origin to a boundary point, $\vQ=\lambda_{\mathrm{rad}}\vQ_{\partial}+(1-\lambda_{\mathrm{rad}})\vQ_0$, with $\vQ_{\partial}\in\partial\cQ$ and $0\le\lambda_{\mathrm{rad}}\le1$, where $\vQ_0$ is the correlation all of whose elements vanish. Applying Theorem~\ref{thm:decomp} to $\vQ_{\partial}$ and expanding its local component into deterministic vertices gives
\begin{equation}\label{eq:radial-decomp}
\vQ=\lambda_0\vQ_{\vtheta}+\sum_{j\in I_{\vtheta}}\lambda_j\vL_j,
\end{equation}
where $\vQ_{\vtheta}\in\cS$, all weights are nonnegative, and $\lambda_0+\sum_{j\in I_{\vtheta}}\lambda_j=\lambda_{\mathrm{rad}}\le1$; the remaining weight $1-\lambda_{\mathrm{rad}}$ is carried by $\vQ_0$ and contributes nothing. If $\vQ_{\partial}$ lies in an identified face $F_I$ from Theorem~\ref{thm:finite-cone}, its local component can be restricted to the vertices indexed by $I$. The classification imposes no two-vertex bound; without an identified face, all sixteen local vertices remain available.

A supplied feasible pair took a median $0.0128\,\mathrm{s}$; cyclic testing of the $28$ pairs took $0.339\,\mathrm{s}$. These are observed numerical-search times, not exact certification costs or worst-case bounds. An unsuccessful local search leaves the input unresolved rather than establishing nonquantumness.

\paragraph*{Summary.}
The quantum boundary of Bell correlations delimits the nonlocality that
nature can produce. Yet even in the minimal Bell scenario, the quantum
set was known only through its extremal points, which cannot by itself
decide whether a given correlation arises in nature: testing it still
requires searching over unknown convex mixtures. Here we settle this
question. A correlation is quantum if and only if it decomposes into a
single pure-two-qubit correlation and a local correlation, reducing the
search to one extremal point. The proof rests on a new uniqueness
result---strict quantum advantage implies a unique quantum
maximizer---and also yields a finite classification of the maximal
proper faces through each nonlocal exposed point. Constructively, the
criterion becomes eight necessary and sufficient inequalities in five
angles, with all mixture weights eliminated and recovered explicitly,
and no separating Bell inequality required. This turns the
extremal-point description into an explicit test of what nature can and
cannot produce.

\onecolumngrid
\appendix
\clearpage

\clearpage
\section{Basic notions}\label{app:basic}
As preparation, we introduce the correlation sets involved here and recall the basic notions of convex geometry.

\subsection{Bell correlations in the minimal Bell scenario}\label{appcontents:A:1}
In the $(2,2,2)$ Bell scenario, Alice and Bob choose inputs $x,y\in\{0,1\}$ in each round and obtain outcomes $a,b\in\{\pm1\}$. A behavior is a family of nonnegative, normalized conditional probabilities $\{p(ab|xy)\}_{a,b,x,y}$. A no-signaling behavior can equivalently be specified by eight expectation values,
\begin{equation*}
\vQ=(A_0,A_1,B_0,B_1,E_{00},E_{01},E_{10},E_{11})\in\mathbb R^8,
\end{equation*}
related to its joint statistics by
\begin{equation}\label{eq:probabilities}
p(ab|xy)=\frac{1+aA_x+bB_y+abE_{xy}}{4}.
\end{equation}
\begin{itemize}
\item $\cL$: the local set---correlations compatible with local realism. It is the convex hull of the sixteen deterministic behaviors $\vL_{\eps;\eta}$, indexed by $\eps=(\eps_0,\eps_1)$ and $\eta=(\eta_0,\eta_1)$ in $\{\pm1\}^2$, and given by
\begin{equation*}
A_x=\eps_x,\qquad B_y=\eta_y,\qquad E_{xy}=\eps_x\eta_y.
\end{equation*}
Thus $\cL$ is a polytope. We use ``local'' and ``classical'' interchangeably.
\item $\cQ$: the quantum set---a behavior $\vQ$ is quantum if it admits a realization
\begin{equation*}
\begin{aligned}
A_x&=\tr\bigl[\rho(\hat A_x\otimes I)\bigr],\\
B_y&=\tr\bigl[\rho(I\otimes\hat B_y)\bigr],\\
E_{xy}&=\tr\bigl[\rho(\hat A_x\otimes\hat B_y)\bigr],
\end{aligned}
\end{equation*}
where $\rho$ is a density operator and $\hat A_x,\hat B_y$ represent the local binary measurements.
\item $\cS$: the pure-two-qubit projective family---the quantum correlations generated by pure two-qubit states and real projective measurements. It is the compact five-parameter family given by
\begin{equation}\label{eq:family-app}
\begin{aligned}
A_x&=\cos\theta\cos a_x,\qquad B_y=\cos\theta\cos b_y,\\
E_{xy}&=\cos a_x\cos b_y+\sin\theta\sin a_x\sin b_y,
\end{aligned}
\end{equation}
where $0\le\theta\le\pi/2$ and all measurement angles are taken modulo $2\pi$. Its members are realized by
\begin{equation*}
\begin{gathered}
|\psi_\theta\rangle=\cos(\theta/2)|00\rangle+\sin(\theta/2)|11\rangle,\\
\hat A_x=\cos a_x Z+\sin a_x X,\qquad\hat B_y=\cos b_y Z+\sin b_y X,
\end{gathered}
\end{equation*}
where $X$ and $Z$ are Pauli matrices \cite{Masanes2005,BarizienBancal2025}.
\item $\Ext(K)$: the extreme-point set of a nonempty compact convex set $K$. A point $\vQ\in K$ is extreme if
\begin{equation*}
\vQ=\lambda \vQ_1+(1-\lambda)\vQ_2,\quad \vQ_1,\vQ_2\in K,\quad0<\lambda<1\ \Longrightarrow\ \vQ_1=\vQ_2=\vQ.
\end{equation*}
Every nonempty compact convex set in finite dimension is generated by its extreme points: $K=\conv\Ext(K)$ \cite{Rockafellar1970}.
\item $\conv E$: the convex hull---a convex mixture $\vQ=\sum_i\lambda_i\vQ_i$, with $\lambda_i\ge0$ and $\sum_i\lambda_i=1$, describes a shared classical random choice among behaviors. The convex hull of $E\subseteq\mathbb R^8$ is the set of all such mixtures,
\begin{equation}\label{eq:conv-def}
\conv E=\left\{\sum_{i=1}^{m}\lambda_i\vQ_i:m\in\mathbb N,\ \vQ_i\in E,\ \lambda_i\ge0,\ \sum_{i=1}^{m}\lambda_i=1\right\}.
\end{equation}
A set is convex if it contains every convex mixture of its points.
\item $\partial K$: the boundary of set $K$.
\end{itemize}
Membership in $\cS$ does not imply extremality, and $\cS$ is nonconvex. Writing $\cE:=\Ext(\cQ)$, we have $\cE\subsetneq\cS\subsetneq\cQ$.
\begin{proposition}[Convex-hull reduction \cite{Masanes2005,BarizienBancal2025}]\label{prop:convex-hull}
The quantum set is compact and
\begin{equation}\label{eq:conv-Q}
\cQ=\conv\cS.
\end{equation}
\end{proposition}
This is a cited input to our argument.

\subsection{Bell functionals and convex geometry}\label{appcontents:A:2}
A Bell functional $\beta_{\vB}$ is a linear functional on behaviors. We identify it with its coefficient vector
\begin{equation*}
\vB=(\alpha_0,\alpha_1,\beta_0,\beta_1,c_{00},c_{01},c_{10},c_{11})\in\mathbb R^8
\end{equation*}
in the dual space, acting on $\vQ$ by
\begin{equation}\label{eq:bell-functional}
\beta_{\vB}(\vQ)=\vB\cdot\vQ=\sum_x\alpha_xA_x+\sum_y\beta_yB_y+\sum_{x,y}c_{xy}E_{xy}.
\end{equation}
Write $C=(c_{xy})_{x,y=0}^{1}$ for the matrix of correlator coefficients. For a nonempty compact correlation set $K$, put
\begin{equation*}
h_K(\vB)=\max_{\vQ\in K}\beta_{\vB}(\vQ),
\end{equation*}
and abbreviate the local and quantum bounds by
\begin{equation*}
L_{\vB}=\max_{\vQ\in\cL}\beta_{\vB}(\vQ),\qquad Q_{\vB}=\max_{\vQ\in\cQ}\beta_{\vB}(\vQ).
\end{equation*}
The inequality $\beta_{\vB}(\vQ)\le L_{\vB}$ is a Bell inequality. A functional has a \emph{strict quantum advantage}, or a strict Bell gap, if $Q_{\vB}>L_{\vB}$.
\begin{itemize}
\item $H_{\vB}$: the \textbf{supporting hyperplane}---for a nonzero Bell functional, the hyperplane supporting a correlation set $K$ is
\begin{equation*}
H_{\vB}=\{\vQ\in\mathbb R^8:\beta_{\vB}(\vQ)=h_K(\vB)\}.
\end{equation*}
Unless another set is specified, $K=\cQ$.
\item $F_{\vB}$: the \textbf{exposed face}---the intersection of $K$ with the corresponding supporting hyperplane,
\begin{equation}\label{eq:exposed-face}
F_{\vB}=K\cap H_{\vB}=\{\vQ\in K:\beta_{\vB}(\vQ)=h_K(\vB)\}.
\end{equation}
\item $\cE_{\vB}$: the \textbf{extreme points of the exposed face},
\begin{equation}\label{eq:exposed-generators}
\cE_{\vB}:=\Ext(F_{\vB}).
\end{equation}
An exposed face is generated by its extreme points. If $F_{\vB}$ consists of a single point, that point is called an \emph{exposed point}. If $|\cE_{\vB}|=n<\infty$, the face has dimension at most $n-1$, with equality when its generators are affinely independent; in that case the face is a simplex. A proper face of $\cQ$ is \emph{maximal} if it is not strictly contained in another proper face.
\item $\cE_{\rm es}$: the set of \textbf{nonlocal exposed points} of $\cQ$.
\item $\scrN(\vQ)$: the \textbf{normal cone} of the quantum correlation set at $\vQ$,
\begin{equation}\label{eq:normal-app}
\scrN(\vQ)=\{\vB\in\mathbb R^8:\vB\cdot(\vQ'-\vQ)\le0\ \text{for every }\vQ'\in\cQ\}.
\end{equation}
Its nonzero elements are precisely the Bell functionals maximized at $\vQ$.
\end{itemize}

\subsection{Bell operators}\label{appcontents:A:3}
\begin{itemize}
\item $\Bell$: the \textbf{Bell operator}---for fixed measurement observables, the functional in \eqref{eq:bell-functional} is represented by
\begin{equation}\label{eq:bell-op-notions}
\Bell=\sum_x\alpha_x\hat A_x\otimes I+\sum_y\beta_y I\otimes\hat B_y+\sum_{x,y}c_{xy}\hat A_x\otimes\hat B_y.
\end{equation}
Its expectation in $\rho$ gives the Bell score. A pair with overlap $s$ is unitarily equivalent to the canonical pair $\hat A_0=Z$ and $\hat A_1=sZ+\sqrt{1-s^2}X$, and similarly for Bob. The corresponding operator is written $\Bell(s,t)$, and Proposition~\ref{prop:convex-hull} reduces the quantum bound to
\begin{equation}\label{eq:bound-notions}
Q_{\vB}=\max_{(s,t)\in[-1,1]^2}\lambda_{\max}\bigl(\Bell(s,t)\bigr).
\end{equation}
\item $s,t$: the \textbf{local measurement overlaps} for traceless projective qubit observables,
\begin{equation*}
s=\tfrac12\tr(\hat A_0\hat A_1),\qquad t=\tfrac12\tr(\hat B_0\hat B_1).
\end{equation*}
Each pair is noncommuting when $|s|<1$ and $|t|<1$, respectively.
\item $M,f$: the \textbf{Bell gap matrix and its determinant}---with $\mu=Q_{\vB}$,
\begin{equation}\label{eq:gap-notions}
M(s,t)=\mu I-\Bell(s,t),\qquad f(s,t)=\det M(s,t).
\end{equation}
The global quantum bound implies $M(s,t)\succeq0$ for all overlaps.
\item $[s^it^j]f$: \textbf{coefficient extraction}---the coefficient of $s^it^j$ in the polynomial $f$.
\end{itemize}
To characterize the maximal proper faces containing a prescribed extreme point, we turn to the extreme rays of its normal cone. The following general convex-geometric result justifies this approach.
\begin{proposition}[Extreme normals and maximal faces]\label{prop:extreme-normals}
Let $K\subset\mathbb R^n$ be a full-dimensional compact convex set and $\vQ\in\partial K$. Every maximal proper face of $K$ containing $\vQ$ can be exposed by a functional generating an extreme ray of $N_K(\vQ)$. The converse need not hold.
\end{proposition}
\begin{proof}
Let $F$ be a maximal proper face containing $\vQ$. Every maximal proper face of a finite-dimensional convex set is exposed \cite{Rockafellar1970}. Hence $F=F_{\vB}$ for some nonzero $\vB\in N_K(\vQ)$.

Choose $\vQ_{\mathrm{ref}}\in\intt K$. Since $\vB\cdot(\vQ-\vQ_{\mathrm{ref}})>0$ for every nonzero $\vB\in N_K(\vQ)$, we may normalize the cone by
\begin{equation*}
N_{\vQ}=\{\vB\in N_K(\vQ):\vB\cdot(\vQ-\vQ_{\mathrm{ref}})=1\}.
\end{equation*}
This is a compact convex section of the normal cone. Therefore, after normalization, $\vB$ admits a finite convex decomposition
\begin{equation*}
\vB=\sum_i\lambda_i \vB_i,\qquad\lambda_i>0,\quad\sum_i\lambda_i=1,\quad \vB_i\in\Ext(N_{\vQ}).
\end{equation*}
The vectors $\vB_i$ generate extreme rays of $N_K(\vQ)$. Since all $\vB_i$ are maximized at $\vQ$,
\begin{equation}\label{eq:face-intersect}
F_{\vB}=\bigcap_i F_{\vB_i}\subseteq F_{\vB_i}.
\end{equation}
Each $F_{\vB_i}$ is a proper face of $K$. By maximality of $F_{\vB}$, we obtain $F_{\vB_i}=F_{\vB}$ for every $i$. Thus every $\vB_i$ exposes the same maximal proper face. The converse fails in general: an extreme normal may expose a face strictly contained in a larger proper face.
\end{proof}
Specialized to $\cQ$, this result shows that the maximal proper faces passing through a nonlocal exposed point $\vQ$ are governed entirely by the extreme rays of the normal cone $\scrN(\vQ)$. Theorem~\ref{thm:maximal-faces} further implies that every such face has $\vQ$ as its unique nonlocal extreme point, all other extreme points being local deterministic behaviors. Every such maximal proper face can therefore be exposed by a Bell functional on an extreme ray of $\scrN(\vQ)$; the converse is not assumed.

\subsection{Frustrated signs}\label{appcontents:A:4}
One is free to relabel the outcomes, which is equivalent to flipping signs in the Bell functional. Flipping Alice's output at input $x$ by $u_x\in\{\pm1\}$ and Bob's output at input $y$ by $v_y\in\{\pm1\}$ gives
\begin{equation*}
A'_x=u_xA_x,\qquad B'_y=v_yB_y,\qquad E'_{xy}=u_xv_yE_{xy}.
\end{equation*}
The coefficients transform accordingly:
\begin{equation}\label{eq:relabel}
\alpha'_x=u_x\alpha_x,\qquad\beta'_y=v_y\beta_y,\qquad c'_{xy}=u_xv_yc_{xy}.
\end{equation}
It preserves the Bell score, $\beta_{\vB'}(\vQ')=\beta_{\vB}(\vQ)$, and both the local and quantum bounds. Input permutations are also score-preserving relabelings.
\begin{lemma}[A local model with fixed marginals]\label{lem:local-coupling}
Given the marginals of any no-signaling behavior, there is a local behavior with the same marginals and correlators
\begin{equation}\label{eq:local-bound}
E_{xy}^{\rm Loc}:=1-|A_x-B_y|\ge E_{xy}.
\end{equation}
If all joint probabilities of the original behavior are strictly positive, the inequality is strict for every input pair $(x,y)$.
\end{lemma}
\begin{proof}
For fixed $(x,y)$, put $r=p(+-|xy)$ and $s'=p(-+|xy)$. Then
\begin{equation*}
E_{xy}=1-2(r+s'),\qquad A_x-B_y=2(r-s').
\end{equation*}
Since $r,s'\ge0$, this gives $E_{xy}\le1-|A_x-B_y|$. More precisely,
\begin{equation}\label{eq:local-difference}
1-|A_x-B_y|-E_{xy}=4\min\{p(+-|xy),p(-+|xy)\}.
\end{equation}
For the local construction, let the parties share a single random variable $U$ uniformly distributed on $[0,1]$. Alice outputs $+1$ when $U\le(1+A_x)/2$, and Bob outputs $+1$ when $U\le(1+B_y)/2$. This reproduces the prescribed marginals. The probability of disagreement is $|A_x-B_y|/2$, so the correlator is $1-|A_x-B_y|$. The same shared variable defines all four input pairs, and hence a local behavior. Strictness follows from \eqref{eq:local-difference}.
\end{proof}

The sign of Bell correlators $C\equiv\{c_{xy}\}_{x,y=0,1}$ plays a central role in the physical motivation and supplies the essential sign structure underlying our uniqueness theorem.
\begin{definition}[Frustrated correlator signs]\label{def:frustration}
A Bell functional has frustrated correlator signs when
\begin{equation}\label{eq:frustration-app}
c_{00}c_{01}c_{10}c_{11}<0.
\end{equation}
\end{definition}
Frustrated signs can be derived from physically motivated properties of the Bell functional although they do not imply nonlocality directly.
\begin{proposition}[Conditions that force frustrated signs]\label{prop:force-frustration}
Each of the following conditions implies \eqref{eq:frustration-app}:
\begin{itemize}
\item[(a)] $Q_{\vB}>L_{\vB}$; the Bell functional presents a quantum advantage;
\item[(b)] $Q_{\vB}=L_{\vB}$ and a maximizing quantum behavior has all joint probabilities strictly positive.
\end{itemize}
\end{proposition}
\begin{proof}
Suppose \eqref{eq:frustration-app} fails. Local output flips can then make every nonzero entry of $C$ positive. To see this, view the four entries as signed edges of the bipartite four-cycle. If an entry is zero, the remaining graph is a forest, whose edge signs can be made positive by vertex flips. If all entries are nonzero, this is possible precisely when the product of their signs is positive.

After these relabelings, $c_{xy}\ge0$ for every $(x,y)$. Lemma~\ref{lem:local-coupling} preserves every marginal contribution and does not decrease any positively weighted correlator. Thus every quantum Bell value is bounded by the value of a local behavior. This rules out a strict Bell gap and proves (a).

For (b), suppose first that $C\ne0$. At a behavior with strictly positive joint probabilities, the same local model strictly increases every correlator, and hence strictly increases the Bell score. This contradicts maximality. If $C=0$, the functional contains only marginals. Its local and quantum maximum is $\sum_x|\alpha_x|+\sum_y|\beta_y|$, attained by deterministic outputs. Strict positivity of the joint probabilities implies $|A_x|,|B_y|<1$. Since $\vB\ne0$, the score at such marginals is strictly below that maximum. This again contradicts maximality.
\end{proof}

\section{Non-trivial maximal proper face}\label{app:faces}\label{sec:unique}
In this section, we prove that a maximal proper face containing one nonlocal exposed point contains no other nonclassical extreme points.

\begin{center}
\begin{minipage}{0.94\textwidth}
\textbf{Proof map for Stage I}\par\smallskip
Frustration + quantum bound + determinant identities $\Longrightarrow f\ge0$, $\deg_s f,\deg_t f\le2$, $[s^2t^2]f<0$.\par
Assume nonuniqueness: two distinct interior zeros $(s,t),(s',t')\in(-1,1)^2$ of $f$.\par
\quad $\bullet$ $s\ne s'$, $t\ne t'$ $\Longrightarrow$ a quartic with two double zeros $\xrightarrow{\text{negative leading coefficient}}$ contradiction with $f\ge0$;\par
\quad $\bullet$ $s=s'$ or $t=t'$ $\Longrightarrow$ an entire zero line $\xrightarrow{\text{global positivity + kernel argument}}$ $M\succ0$ on its interior: contradiction.\par
Hence at most one optimal interior overlap pair; a simple top eigenvalue (independent eigenspace lemma) gives at most one nonlocal maximizing member of $\cS$. Then $\cQ=\conv\cS\Rightarrow$ at most one nonlocal extreme point in $F_{\vB}$; and $Q_{\vB}>L_{\vB}\Rightarrow$ no local maximizer $\Rightarrow F_{\vB}=\{\vQ_{\vtheta}\}$ as a byproduct.
\end{minipage}
\end{center}

\subsection{Determinant polynomial of  Bell operator}\label{appcontents:B:1}

The Bell functional acts on behaviors. To optimize it over quantum realizations, we now represent the same score by an operator and reduce the measurement dependence to two real parameters.

For fixed observables, the Bell operator associated with \eqref{eq:bell-functional} is
\begin{equation}\label{eq:bell-op}
\Bell=\sum_x\alpha_x\hat A_x\otimes I+\sum_y\beta_y I\otimes\hat B_y+\sum_{x,y}c_{xy}\hat A_x\otimes\hat B_y.
\end{equation}
The score of a realization is $\beta_{\vB}(\vQ)=\tr(\rho\Bell)$. For these fixed measurements, maximizing over states gives $\lambda_{\max}(\Bell)$.

\subsubsection{Canonical measurements and the overlap boundary}
For nontrivial projective qubit observables, define
\begin{equation}\label{eq:overlaps}
s=\tfrac12\tr(\hat A_0\hat A_1),\qquad t=\tfrac12\tr(\hat B_0\hat B_1).
\end{equation}
Two such observables are determined up to a unitary by their overlap. We may therefore choose
\begin{equation}\label{eq:canonical}
\begin{aligned}
\hat A_0&=Z,&\hat A_1&=sZ+\sqrt{1-s^2}X,\\
\hat B_0&=Z,&\hat B_1&=tZ+\sqrt{1-t^2}X,
\end{aligned}
\end{equation}
with $(s,t)\in[-1,1]^2$. The Alice pair is noncommuting precisely when $|s|<1$, and similarly the Bob pair is noncommuting precisely when $|t|<1$.

Substituting \eqref{eq:canonical} into \eqref{eq:bell-op} gives the operator $\Bell(s,t)$. This is a real symmetric matrix, so a maximizing eigenvector can be chosen real. By Proposition~\ref{prop:convex-hull}, the global quantum bound is
\begin{equation}\label{eq:global-bound}
Q_{\vB}=\max_{(s,t)\in[-1,1]^2}\lambda_{\max}\bigl(\Bell(s,t)\bigr).
\end{equation}

\subsection{Excluding two quantum maximizers  in the case $s\neq s'$ or $t\neq t'$.}\label{secline}
Fix the functional and put $\mu=Q_{\vB}$. We do not assume that $\mu>L_{\vB}$. Define
\begin{equation}\label{eq:gap-matrix}
M(s,t)=\mu I-\Bell(s,t),\qquad f(s,t)=\det M(s,t).
\end{equation}
Because $\mu$ is optimized over all states and all measurements,
\begin{equation}\label{eq:global-psd}
M(s,t)\succeq0\qquad\text{for every }(s,t)\in[-1,1]^2.
\end{equation}
The same coefficients and the same value of $\mu$ are used throughout the overlap square. Therefore
\begin{equation}\label{eq:det-zero-criterion}
f(s,t)\ge0,\qquad f(s,t)=0\iff\lambda_{\max}(\Bell(s,t))=\mu.
\end{equation}
Indeed, all eigenvalues of $\Bell(s,t)$ are at most $\mu$, so the determinant vanishes precisely when one of them equals $\mu$. The zeros of $f$ thus record optimal overlap pairs. The next section identifies the algebraic structure of this scalar function.

Although the canonical Bell operator contains square roots, its determinant is a polynomial in $s$ and $t$. We shall derive its separate degree bounds and identify the one coefficient whose sign will matter:
\begin{equation*}
\deg_s f,\deg_t f\le2,\qquad[s^2t^2]f=16c_{00}c_{01}c_{10}c_{11}.
\end{equation*}
Here $[s^it^j]f$ denotes coefficient extraction. Under frustration, the last coefficient is negative.

\subsubsection{A Pauli-trace calculation}
Let $\sigma=(Z,X)$ and introduce
\begin{equation}\label{eq:phi-gram}
\Phi_s=\begin{pmatrix}1&s\\0&\sqrt{1-s^2}\end{pmatrix},\qquad G_s=\Phi_s^T\Phi_s=\begin{pmatrix}1&s\\s&1\end{pmatrix},
\end{equation}
with analogous definitions for $t$. Set
\begin{equation*}
\begin{gathered}
\mathbf a=(\alpha_0,\alpha_1)^T,\quad\mathbf b=(\beta_0,\beta_1)^T,\\
\widetilde{\mathbf a}=\Phi_s\mathbf a,\quad\widetilde{\mathbf b}=\Phi_t\mathbf b,\quad \widetilde C=\Phi_s C\Phi_t^T.
\end{gathered}
\end{equation*}
Then
\begin{equation}\label{eq:pauli-op}
\Bell(s,t)=(\widetilde{\mathbf a}\cdot\sigma)\otimes I+I\otimes(\widetilde{\mathbf b}\cdot\sigma)+\sum_{i,j=1}^2\widetilde C_{ij}\sigma_i\otimes\sigma_j.
\end{equation}
Define
\begin{equation}\label{eq:wdr}
\begin{aligned}
w&=\|\widetilde{\mathbf a}\|^2+\|\widetilde{\mathbf b}\|^2+\|\widetilde C\|_F^2,\\
d&=\widetilde{\mathbf a}^{T}\widetilde C\widetilde{\mathbf b},\\
r&=\|\widetilde C^T\widetilde{\mathbf a}\|^2+\|\widetilde C\widetilde{\mathbf b}\|^2+\|\widetilde{\mathbf a}\|^2\|\widetilde{\mathbf b}\|^2+(\det\widetilde C)^2,
\end{aligned}
\end{equation}
where $\|\cdot\|_F$ is the Frobenius norm.
\begin{lemma}[Trace identities]\label{lem:trace}
The Bell operator satisfies
\begin{equation}\label{eq:traces}
\tr\Bell=0,\quad\tr\Bell^2=4w,\quad\tr\Bell^3=24d,\quad\tr\Bell^4=4w^2+16r.
\end{equation}
\end{lemma}
\begin{proof}
Use $ZX=iY$, $XZ=-iY$, and the orthogonality of Pauli tensors. Direct multiplication gives $\Bell^2=wI+2\Delta$, where
\begin{equation}\label{eq:delta}
\Delta=(\widetilde C\widetilde{\mathbf b})\cdot\sigma\otimes I+I\otimes(\widetilde C^T\widetilde{\mathbf a})\cdot\sigma+(\widetilde{\mathbf a}\cdot\sigma)\otimes(\widetilde{\mathbf b}\cdot\sigma)-(\det\widetilde C)Y\otimes Y.
\end{equation}
These four terms are pairwise orthogonal in the Hilbert--Schmidt inner product. Hence $\tr\Delta=0$, $\tr\Delta^2=4r$, and
\begin{equation*}
\tr\Bell^4=\tr(wI+2\Delta)^2=4w^2+16r.
\end{equation*}
For the cubic trace, each of the three matching Pauli pairings in $\tr(\Bell\Delta)$ contributes $4\widetilde{\mathbf a}^T\widetilde C\widetilde{\mathbf b}$. Therefore $\tr\Bell^3=2\tr(\Bell\Delta)=24d$. The remaining identities follow directly from Pauli orthogonality.
\end{proof}
Newton's identities for a traceless $4\times4$ matrix give
\begin{equation*}
\det(\mu I-\Bell)=\mu^4-\frac{\tr\Bell^2}{2}\mu^2-\frac{\tr\Bell^3}{3}\mu+\frac{(\tr\Bell^2)^2}{8}-\frac{\tr\Bell^4}{4}.
\end{equation*}
Substituting \eqref{eq:traces} yields
\begin{equation}\label{eq:det-expression}
f=\mu^4-2w\mu^2-8d\mu+w^2-4r.
\end{equation}

\subsubsection{Removing the square roots and isolating the leading coefficient}
The square-root coordinates disappear when the three scalar quantities are written in terms of Gram matrices:
\begin{equation}\label{eq:gram-wdr}
\begin{aligned}
w={}&\mathbf a^T G_s\mathbf a+\mathbf b^TG_t\mathbf b+\tr(C^T G_s C G_t),\\
d={}&\mathbf a^T G_s C G_t\mathbf b,\\
r={}&\mathbf a^T G_s C G_t C^T G_s\mathbf a+\mathbf b^T G_t C^T G_s C G_t\mathbf b\\
 &+(\mathbf a^T G_s\mathbf a)(\mathbf b^T G_t\mathbf b)+(\det C)^2(1-s^2)(1-t^2).
\end{aligned}
\end{equation}
\begin{proposition}[Polynomial structure]\label{prop:polynomial}
The function $f(s,t)$ is a polynomial with
\begin{equation}\label{eq:polynomial-coeff}
\deg_s f,\deg_t f\le2,\qquad[s^2t^2]f=16c_{00}c_{01}c_{10}c_{11}.
\end{equation}
\end{proposition}
\begin{proof}
Each Gram matrix is affine in its own variable. Thus $w$ and $d$ have separate degree at most one, and every term of $r$ has separate degree at most two. Equation~\eqref{eq:det-expression} proves the stated degree bounds.

To compute the coefficient, write $G_s=I+sJ$ and $G_t=I+tJ$, where $J=\begin{psmallmatrix}0&1\\1&0\end{psmallmatrix}$. Then
\begin{equation*}
[st]w=\tr(C^TJCJ)=2(c_{00}c_{11}+c_{01}c_{10}).
\end{equation*}
The first three terms of $r$ in \eqref{eq:gram-wdr} have no $s^2t^2$ term: their bidegrees are at most $(2,1)$, $(1,2)$, and $(1,1)$, respectively. The last term contributes $(\det C)^2$. Consequently,
\begin{equation}\label{eq:highest-coeff}
\begin{aligned}
[s^2t^2]f&=4(c_{00}c_{11}+c_{01}c_{10})^2-4(c_{00}c_{11}-c_{01}c_{10})^2\\
&=16c_{00}c_{01}c_{10}c_{11}.
\end{aligned}
\end{equation}
\end{proof}
Non-negativity of $f$ comes from the global operator bound, and frustration fixes the sign of its highest mixed coefficient. These facts will constrain the zeros without requiring an explicit formula for the largest eigenvalue.

We treat a frustrated Bell functional without assuming a strict Bell gap. The stronger conclusion is needed to include maximizing faces that contain local behaviors; strict-gap uniqueness will follow as a special case.

\subsubsection{Statement and proof strategy}
Fix $\vB$ with $c_{00}c_{01}c_{10}c_{11}<0$, and use the notation $\mu=Q_{\vB}$, $\Bell(s,t)$, $M(s,t)$, and $f(s,t)$ from Section~\ref{app:basic}. Recall that $F_{\vB}$ is the quantum maximizing face in \eqref{eq:exposed-face}.
\begin{theorem}[Structure of non-trivial maximizing face]\label{thm:frustrated-face}
There is at most one optimal overlap pair in $(-1,1)^2$. At that pair, if it exists, the optimal state projector is unique. The face $F_{\vB}$ contains at most one nonlocal exposed extreme point. More precisely, if $F_{\vB}$ contains a nonlocal behavior, there is a unique nonlocal $\vQ_{\vtheta}\in F_{\vB}\cap\cS$, and
\begin{equation}\label{eq:face-structure}
F_{\vB}=\conv\bigl(\{\vQ_{\vtheta}\}\cup(F_{\vB}\cap\cL)\bigr).
\end{equation}
Moreover, $\vQ_{\vtheta}$ is an extreme point in $\cQ$. In particular, if $Q_{\vB}>L_{\vB}$, then $F_{\vB}=\{\vQ_{\vtheta}\}$.
\end{theorem}
The proof separates two sources of possible nonuniqueness: different states for the same measurements, and different measurement pairs. Section~\ref{sec:unique} excludes the former. Section~\ref{sec:unique} reduces the latter to a horizontal or vertical line of optimal overlaps. Section~\ref{sec:zero-line} excludes this line using the global Bell bound. Section~\ref{sec:full-face} then passes from qubit realizations to the whole maximizing face.

The inputs already established are
\begin{equation}\label{eq:stage1-inputs}
\begin{gathered}
\cQ=\conv\cS,\quad M(s,t)\succeq0,\quad f(s,t)\ge0,\\
\deg_s f,\deg_t f\le2,\quad[s^2t^2]f=16c_{00}c_{01}c_{10}c_{11}<0.
\end{gathered}
\end{equation}
The negative coefficient is needed only in the polynomial argument. The state-simplicity and zero-line lemmas require only that all $c_{xy}$ be nonzero. Throughout, a zero of $f$ is an optimum by \eqref{eq:det-zero-criterion}.

\subsubsection{Fixed optimal overlaps determine a unique state}
The key observation is that a degenerate two-qubit eigenspace contains a product vector. At an interior optimum, the resulting local behavior would force a correlator coefficient to vanish.
\begin{lemma}[A product vector in a two-dimensional subspace]\label{lem:product-subspace}
Every complex subspace of $\mathbb C^2\otimes\mathbb C^2$ of dimension at least two contains a nonzero product vector.
\end{lemma}
\begin{proof}
Choose two independent vectors and identify their coefficients with matrices $W_0,W_1\in M_2(\mathbb C)$. A nonzero vector is a product vector exactly when its coefficient matrix has rank one. The homogeneous quadratic
\begin{equation*}
\det(\zeta_0W_0+\zeta_1W_1)
\end{equation*}
either vanishes identically or has a zero on the complex projective line. At any such projective zero, independence of $W_0,W_1$ ensures that the matrix is nonzero. Its determinant is zero, so it has rank one.
\end{proof}
\begin{lemma}[Simplicity at an interior optimum]\label{lem:simple-eigenvalue}
Suppose every $c_{xy}\ne0$. If $(s,t)\in(-1,1)^2$ and $\lambda_{\max}(\Bell(s,t))=\mu$, then the eigenvalue $\mu$ is simple.
\end{lemma}
\begin{proof}
Suppose instead that the top eigenspace has dimension at least two. Lemma~\ref{lem:product-subspace} supplies a normalized product state $|u\rangle\otimes|v\rangle$ in that eigenspace. Its behavior is local and has score $\mu$, so
\begin{equation}\label{eq:local-equal}
\mu\le L_{\vB}\le Q_{\vB}=\mu.
\end{equation}
Thus $\mu=L_{\vB}$. This already contradicts a strict Bell gap; moreover, the noncommutativity of the measurements and the separability of the product state also exclude equality when every $c_{xy}\ne0$. To see this, put $a_x=\langle u|\hat A_x|u\rangle$ and $b_y=\langle v|\hat B_y|v\rangle$. The product behavior is a convex mixture of deterministic behaviors with independent Bernoulli distributions for all four outcomes. Because each pair of observables is noncommuting, there is at least one setting $x$ with $|a_x|<1$ and one setting $y$ with $|b_y|<1$. Fixing all other deterministic assignments, all four choices of the outputs at $x,y$ occur with positive weight. Since their mixture reaches the local bound, each of those four assignments must attain the same maximum. The alternating sum of their four Bell scores is $4c_{xy}$, so $c_{xy}=0$, a contradiction.
\end{proof}
Two noncommuting Hermitian qubit observables cannot share an eigenvector: its orthogonal complement would also be invariant under both, giving a common eigenbasis. An optimal density operator must be supported on the top eigenspace. Lemma~\ref{lem:simple-eigenvalue} therefore makes the state projector unique for fixed optimal interior overlaps. Since the Bell operator is real symmetric, its top eigenvector can be chosen real. Equal overlaps give locally unitarily equivalent ordered measurement pairs, so they also give the same optimal behavior. This is not a claim of uniqueness for arbitrary-dimensional realizations.

\subsubsection{Two distinct optimal overlap pairs would force a zero line}
\begin{lemma}[Polynomial reduction to an axis-parallel zero line]\label{lem:zero-line-poly}
Let $f\ge0$ on $[-1,1]^2$, with $\deg_s f,\deg_t f\le2$ and $[s^2t^2]f<0$. If $f$ has two distinct interior zeros, they share one coordinate, and $f$ vanishes identically along the corresponding horizontal or vertical line.
\end{lemma}
\begin{proof}
Let $z_0=(s_0,t_0)$ and $z_1=(s_1,t_1)$ be the two zeros, and restrict $f$ to the line through them:
\begin{equation*}
g(u)=f\bigl(z_0+u(z_1-z_0)\bigr).
\end{equation*}
Because both zeros are interior, $g$ is nonnegative in neighborhoods of $u=0$ and $u=1$. Both are therefore stationary zeros:
\begin{equation}\label{eq:double-zeros}
g(0)=g'(0)=g(1)=g'(1)=0.
\end{equation}
If both coordinates differ, $g$ has degree four with leading coefficient
\begin{equation*}
\kappa=[s^2t^2]f\,(s_1-s_0)^2(t_1-t_0)^2<0.
\end{equation*}
Because the separate degrees are at most two, $s^2t^2$ is the only monomial of total degree four. Every other monomial contributes at most a cubic term in $u$. The two double zeros in \eqref{eq:double-zeros} therefore force
\begin{equation*}
g(u)=\kappa u^2(u-1)^2<0\qquad(0<u<1),
\end{equation*}
contradicting nonnegativity on the segment.

If one coordinate is shared, $g$ has degree at most two. Its derivative has degree at most one and vanishes at both $0$ and $1$, so $g'\equiv0$. Since $g(0)=0$, we obtain $g\equiv0$. The parameter $u$ covers the entire horizontal or vertical line.

An identically zero line is compatible with a nonnegative determinant. Thus Lemma~\ref{lem:zero-line-poly} is a reduction, not yet a contradiction. To exclude that remaining possibility, we must return to the operator inequality \eqref{eq:global-psd}.
\end{proof}

\subsection{Excluding two quantum maximizers in the cases $s=s'$ or $t=t'$.}\label{sec:zero-line}
\begin{lemma}[No axis-parallel line of optima]\label{lem:no-zero-line}
Suppose every $c_{xy}\ne0$. Neither of the following is possible:
\begin{equation*}
f(s,t_*)\equiv0\ \text{with }|t_*|<1,\qquad f(s_*,t)\equiv0\ \text{with }|s_*|<1.
\end{equation*}
\end{lemma}
\begin{proof}
It suffices to exclude a horizontal line. Suppose $f(s,t_*)\equiv0$ for some $|t_*|<1$. By \eqref{eq:det-zero-criterion},
\begin{equation}\label{eq:zero-line-claim}
\lambda_{\max}(\Bell(s,t_*))=\mu\qquad(-1<s<1).
\end{equation}
We proceed in four steps. The endpoints determine coefficient relations; those relations yield a useful decomposition of $M$; a kernel test then forces a Bob operator to be positive for every overlap; this finally makes $M$ strictly positive on the proposed zero line.

\textbf{Step 1. Use the endpoints to determine the coefficient relations.} Continuity extends \eqref{eq:zero-line-claim} to $s=\pm1$. At these endpoints, $\hat A_1=\pm\hat A_0$, so the Bell operator is block diagonal in Alice's $Z$ basis. For an Alice assignment $\eps=(\eps_0,\eps_1)$, its Bob block is
\begin{equation}\label{eq:bob-block}
\Bell_\eps=(\boldsymbol\alpha\cdot\eps)I+e_0(\eps)\hat B_0+e_1(\eps)\hat B_1,\quad e_y(\eps)=\beta_y+c_{0y}\eps_0+c_{1y}\eps_1,\quad\boldsymbol\alpha=(\alpha_0,\alpha_1).
\end{equation}
At $s=1$, the two blocks correspond to $(+,+)$ and $(-,-)$; at $s=-1$, they correspond to $(+,-)$ and $(-,+)$. Choose an optimal block at each endpoint, with assignments
\begin{equation*}
\eps^+\in\{(+,+),(-,-)\},\qquad \eps^-\in\{(+,-),(-,+)\}.
\end{equation*}
The superscripts label endpoints; the selected block need not be unique.

For either selected assignment, abbreviate $e_y=e_y(\eps)$. Since the Bob observables are traceless and have overlap $t_*$,
\begin{equation*}
(e_0\hat B_0+e_1\hat B_1)^2=(e_0^2+e_1^2+2t_*e_0e_1)I.
\end{equation*}
Comparing the optimal block value with the best classical Bob outputs therefore gives
\begin{equation}\label{eq:block-upper}
\mu=\boldsymbol\alpha\cdot\eps+\sqrt{e_0^2+e_1^2+2t_*e_0e_1}\le\boldsymbol\alpha\cdot\eps+|e_0|+|e_1|\le L_{\vB}\le\mu.
\end{equation}
All inequalities are equalities. Hence $\mu=L_{\vB}$ and $e_0e_1=0$ for each selected block. Indeed, when $e_0e_1\ne0$,
\begin{equation*}
(|e_0|+|e_1|)^2-(e_0^2+e_1^2+2t_*e_0e_1)=2|e_0e_1|\bigl(1-t_*\sgn(e_0e_1)\bigr)>0,
\end{equation*}
which would make the first inequality strict. In the strict-gap case, this already excludes the line. We continue to handle $\mu=L_{\vB}$.

The assignments $\eps^+$ and $\eps^-$ differ at exactly one Alice input $x$. Their effective coefficients satisfy
\begin{equation}\label{eq:effective-difference}
e_y(\eps^-)-e_y(\eps^+)=-2\eps_x^+c_{xy}\ne0\qquad(y=0,1).
\end{equation}
Thus the two rows cannot have a zero in the same column. Neither row can have two zeros, since the other would then have none. Each row therefore has exactly one zero, and these zeros occur in different columns.

Apply one common relabeling of the Bell problem: label the differing Alice input by $1$, flip Alice's outputs to make the two selected assignments $(+,+)$ and $(+,-)$, and then exchange Bob's inputs or flip his outputs as needed. The effective coefficients take the form
\begin{equation}\label{eq:effective-table}
\begin{array}{c|cc}
\text{Alice assignment}&e_0&e_1\\\hline
(+,+)&2p&0\\
(+,-)&0&2q
\end{array}\qquad p,q>0.
\end{equation}
Here $p,q$ are coefficient parameters, not probabilities. The score-preserving relabelings of Section~\ref{app:basic} preserve both bounds, the nonvanishing of the $c_{xy}$, and the interior of the overlap square. We use the same notation for the relabeled coefficients and for the relabeled interior overlap $t_*$.

The zero entries give
\begin{equation*}
\beta_1+c_{01}+c_{11}=0,\qquad\beta_0+c_{00}-c_{10}=0.
\end{equation*}
The nonzero entries then yield $2p=2c_{10}$ and $2q=-2c_{11}$. Since both rows are optimal,
\begin{equation*}
\mu=\alpha_0+\alpha_1+2p=\alpha_0-\alpha_1+2q.
\end{equation*}
We have therefore obtained the fixed coefficient relations
\begin{equation}\label{eq:coefficient-relations}
\begin{gathered}
c_{10}=p>0,\quad c_{11}=-q<0,\quad\alpha_1=q-p,\\
\beta_0=p-c_{00},\quad\beta_1=q-c_{01},\quad\mu=\alpha_0+p+q.
\end{gathered}
\end{equation}
Once derived, these relations hold when either party's measurements are varied.

\textbf{Step 2. Decompose the Bell gap into a positive term and a Bob term.} Substitution of \eqref{eq:coefficient-relations} into the Bell operator and gap matrix, Eqs.~\eqref{eq:bell-op} and \eqref{eq:gap-matrix}, gives
\begin{equation}\label{eq:gap-decomp}
M=\Pi+(I-\hat A_0)\otimes D,
\end{equation}
where
\begin{equation}\label{eq:pi-D}
\begin{aligned}
\Pi&=p(I+\hat A_1)\otimes(I-\hat B_0)+q(I-\hat A_1)\otimes(I-\hat B_1),\\
D&=\alpha_0I+c_{00}\hat B_0+c_{01}\hat B_1.
\end{aligned}
\end{equation}
For verification, expand $\Pi$ as
\begin{equation*}
\begin{split}
\Pi={}&(p+q)I+(p-q)\hat A_1\otimes I-pI\otimes\hat B_0-qI\otimes\hat B_1\\
&-p\hat A_1\otimes\hat B_0+q\hat A_1\otimes\hat B_1.
\end{split}
\end{equation*}
Adding $(I-\hat A_0)\otimes D$ supplies the remaining $\hat A_0$ terms and changes the Bob marginal coefficients to $-\beta_0,-\beta_1$. The scalar coefficient becomes $\alpha_0+p+q=\mu$, verifying \eqref{eq:gap-decomp}.

Introduce the spectral projectors $P^{\hat A}_{x,\pm}=(I\pm\hat A_x)/2$ and $P^{\hat B}_{y,\pm}=(I\pm\hat B_y)/2$. Then
\begin{equation}\label{eq:pi-projectors}
\Pi=4pP^{\hat A}_{1,+}\otimes P^{\hat B}_{0,-}+4qP^{\hat A}_{1,-}\otimes P^{\hat B}_{1,-}\succeq0.
\end{equation}
We do not yet know that $D\succeq0$; this cannot be concluded by subtracting $\Pi\succeq0$ from $M\succeq0$.

\textbf{Step 3. Test the kernel of $\Pi$ to prove $D(t)\succeq0$ for every interior $t$.} Choose arbitrary noncommuting measurement pairs on both sides. Let $|u_{x,\pm}\rangle$ and $|v_{y,\pm}\rangle$ be their normalized local eigenvectors. Writing
\begin{equation*}
|\psi\rangle=|u_{1,+}\rangle\otimes|\xi_+\rangle+|u_{1,-}\rangle\otimes|\xi_-\rangle,
\end{equation*}
we obtain from \eqref{eq:pi-projectors}
\begin{equation*}
\langle\psi|\Pi|\psi\rangle=4p\|P^{\hat B}_{0,-}\xi_+\|^2+4q\|P^{\hat B}_{1,-}\xi_-\|^2.
\end{equation*}
Since $p,q>0$,
\begin{equation}\label{eq:pi-kernel}
\ker\Pi=\operatorname{span}\{\,|u_{1,+}\rangle\otimes|v_{0,+}\rangle,\ |u_{1,-}\rangle\otimes|v_{1,+}\rangle\,\}.
\end{equation}
Consider the map
\begin{equation}\label{eq:theta-map}
\Theta:\ker\Pi\longrightarrow\mathbb C^2,\qquad\Theta\psi=(\langle u_{0,-}|\otimes I)\psi.
\end{equation}
On the two displayed basis vectors, its values are
\begin{equation*}
\langle u_{0,-}|u_{1,+}\rangle|v_{0,+}\rangle,\qquad\langle u_{0,-}|u_{1,-}\rangle|v_{1,+}\rangle.
\end{equation*}
Alice's noncommutation makes both scalar factors nonzero: a zero factor would imply a common eigenvector of her two observables. Bob's noncommutation makes $|v_{0,+}\rangle$ and $|v_{1,+}\rangle$ linearly independent, since otherwise they would define a common eigenvector. Thus the two basis vectors of $\ker\Pi$ have linearly independent images, and $\Theta$ is an isomorphism.

For any $|\chi\rangle\in\mathbb C^2$, choose $|\psi\rangle\in\ker\Pi$ with $\Theta\psi=\chi$. Since $I-\hat A_0=2|u_{0,-}\rangle\langle u_{0,-}|$, the global inequality \eqref{eq:global-psd} gives
\begin{equation}\label{eq:kernel-positive}
0\le\langle\psi|M|\psi\rangle=2\langle\chi|D|\chi\rangle.
\end{equation}
Thus $D\succeq0$. The Bob pair was arbitrary, so this conclusion holds for every $-1<t<1$, not merely for $t=t_*$. This is the essential use of the global quantum bound.

\textbf{Step 4. Return to $t_*$ and obtain strict positivity.} For a Bob pair of overlap $t$, the eigenvalues of $D$ are
\begin{equation*}
\alpha_0\pm r(t),\qquad r(t)=\sqrt{c_{00}^2+c_{01}^2+2c_{00}c_{01}t}.
\end{equation*}
Positivity for all $-1<t<1$ forces
\begin{equation}\label{eq:alpha-bound}
\alpha_0\ge\sup_{-1<t<1}r(t)=|c_{00}|+|c_{01}|.
\end{equation}
Since $c_{00}c_{01}\ne0$ and $|t_*|<1$,
\begin{equation*}
(|c_{00}|+|c_{01}|)^2-r(t_*)^2=2|c_{00}c_{01}|\bigl(1-t_*\sgn(c_{00}c_{01})\bigr)>0.
\end{equation*}
Hence $\alpha_0>r(t_*)$, so $D(t_*)\succ0$.

Fix any interior $s$, and suppose $M(s,t_*)\psi=0$. The two terms in \eqref{eq:gap-decomp} are now positive semidefinite, so
\begin{equation*}
0=\langle\psi|\Pi|\psi\rangle+\langle\psi|(I-\hat A_0)\otimes D(t_*)|\psi\rangle
\end{equation*}
forces each expectation to vanish. The first gives $\psi\in\ker\Pi$. The second gives $2\langle\Theta\psi|D(t_*)|\Theta\psi\rangle=0$. Strict positivity of $D(t_*)$ gives $\Theta\psi=0$, and injectivity of $\Theta$ on $\ker\Pi$ gives $\psi=0$. Therefore
\begin{equation}\label{eq:zero-line-contradiction}
M(s,t_*)\succ0,\qquad f(s,t_*)>0\quad(-1<s<1),
\end{equation}
contradicting the proposed zero line. Exchanging Alice and Bob excludes a vertical line as well.
\end{proof}
\begin{remark}[What the endpoint argument does and does not assert]\label{rem:endpoints}
The normal form \eqref{eq:effective-table} identifies four deterministic strategies attaining $\mu=L_{\vB}$:
\begin{equation*}
(+,+;+,+),\quad(+,+;+,-),\quad(+,-;+,+),\quad(+,-;-,+).
\end{equation*}
They need not be the only local maximizers, and they need not be individually realizable using Bob's fixed noncommuting observables. For example, the first optimal block is attained by the $+1$ eigenstate of $\hat B_0$. Its endpoint behavior is
\begin{equation*}
\vL_+=\lambda \vL_{(+,+;+,+)}+(1-\lambda)\vL_{(+,+;+,-)},\qquad\lambda=\frac{1+t_*}{2}\in(0,1).
\end{equation*}
The second block analogously mixes the other two strategies with the same weights. Bob's unused effective coefficient is zero, so the remaining measurement uncertainty does not alter the score. The commuting endpoints are used only to constrain the fixed coefficients. Steps 3 and 4 concern noncommuting measurement pairs; no interior target is replaced by a commuting realization.
\end{remark}

\subsection{From uniqueness theorem to the full maximizing face}\label{sec:full-face}

This consequence uses one additional cited input, separate from the convex-hull reduction and from the determinant argument: nonlocal exposed points of $\cQ$ have all joint probabilities strictly positive \cite{BarizienBancal2025}. We use this property as an external input; we do not derive it here. It is needed only for Theorem~\ref{thm:no-common-support}, not for the decomposition theorem.
\begin{theorem}[No common support]\label{thm:no-common-support}
If $\vQ_1\ne \vQ_2$ belong to the set $\cE_{\rm es}$ of nonlocal exposed points, then no nonzero Bell functional is maximized at both points. Equivalently,
\begin{equation*}
\scrN(\vQ_1)\cap\scrN(\vQ_2)=\{0\}.
\end{equation*}
\end{theorem}
\begin{proof}
Suppose a nonzero $\vB$ is maximized at both $\vQ_1$ and $\vQ_2$, where $\vQ_1,\vQ_2\in\cE_{\rm es}$. By the cited input, $\vQ_1$ has strictly positive joint probabilities. Therefore Proposition~\ref{prop:force-frustration}(b) forces frustrated signs, even if $Q_{\vB}=L_{\vB}$. Each $\vQ_i$ is a nonlocal extreme point of $\cQ$ and hence a nonlocal extreme point of the face $F_{\vB}$. Theorem~\ref{thm:frustrated-face} allows at most one such point, so $\vQ_1=\vQ_2$. This proves the supporting-hyperplane statement. Equation~\eqref{eq:normal-app} gives its equivalent normal-cone formulation.
\end{proof}

\begin{proof}[Proof of Theorem~\ref{thm:frustrated-face}]
We first obtain uniqueness at the qubit level. By \eqref{eq:stage1-inputs} and Lemma~\ref{lem:zero-line-poly}, two distinct optimal interior overlap pairs would force a zero line. Lemma~\ref{lem:no-zero-line} excludes that line. Thus there is at most one optimal pair in $(-1,1)^2$. At that pair, if it exists, Lemma~\ref{lem:simple-eigenvalue} gives a unique optimal state projector. Local unitary equivalence of ordered pairs with the same overlaps then gives a unique optimal behavior at that pair.

By Lemma 7.1 of Ref.~\cite{Goh2018}, every nonlocal member of $\cS$ has interior overlaps. Thus the preceding qubit argument covers every nonlocal maximizing member of $\cS$, and all such members coincide.

Now take $\vQ\in F_{\vB}$. The convex-hull input supplies a finite decomposition
\begin{equation*}
\vQ=\sum_j\lambda_j\vQ_{\vtheta_j},\qquad \vQ_{\vtheta_j}\in\cS,\quad \lambda_j>0,\quad\sum_j\lambda_j=1.
\end{equation*}
Saturation of the Bell bound gives
\begin{equation}\label{eq:saturation-sum}
0=\mu-\beta_{\vB}(\vQ)=\sum_j\lambda_j\bigl[\mu-\beta_{\vB}(\vQ_{\vtheta_j})\bigr].
\end{equation}
Every summand is nonnegative, so each $\vQ_{\vtheta_j}$ also belongs to $F_{\vB}$. If $\vQ$ is nonlocal, at least one component is nonlocal; call it $\vQ_{\vtheta}$. All nonlocal components equal $\vQ_{\vtheta}$, and all remaining components lie in $F_{\vB}\cap\cL$. This proves the inclusion of $F_{\vB}$ in the right-hand side of \eqref{eq:face-structure}; the reverse inclusion follows from convexity.

It remains to show that $\vQ_{\vtheta}$ is extreme in $\cQ$. Suppose
\begin{equation*}
\vQ_{\vtheta}=\lambda_0 \vQ_1+(1-\lambda_0)\vQ_2,\qquad \vQ_1,\vQ_2\in\cQ,\quad0<\lambda_0<1.
\end{equation*}
Saturation forces $\vQ_1,\vQ_2\in F_{\vB}$. If $F_{\vB}\cap\cL=\varnothing$, Eq.~\eqref{eq:face-structure} already makes the face a singleton. Otherwise, write
\begin{equation*}
\vQ_i=\lambda_i\vQ_{\vtheta}+(1-\lambda_i)\vL^{(i)},\qquad \vL^{(i)}\in F_{\vB}\cap\cL,\quad0\le\lambda_i\le1.
\end{equation*}
Set $\lambda=\lambda_0\lambda_1+(1-\lambda_0)\lambda_2$. If $\lambda<1$, substitution yields $\vQ_{\vtheta}=\lambda \vQ_{\vtheta}+(1-\lambda)\vL$ for a local behavior $\vL$, which would imply $\vQ_{\vtheta}=\vL$. This contradicts nonlocality. Hence $\lambda=1$, forcing $\lambda_1=\lambda_2=1$ and $\vQ_1=\vQ_2=\vQ_{\vtheta}$. Thus $\vQ_{\vtheta}$ is extreme. Every other nonlocal point in the face is a nontrivial mixture of $\vQ_{\vtheta}$ and a local point, and is therefore not extreme.

Finally, if $\mu>L_{\vB}$, no local behavior belongs to $F_{\vB}$. Since this face is nonempty, Eq.~\eqref{eq:face-structure} reduces to $F_{\vB}=\{\vQ_{\vtheta}\}$.
\end{proof}
The distinction between the two regimes is therefore explicit: a strict Bell gap gives a unique quantum maximizer, whereas equality of the two bounds allows local maximizing components and their mixtures with the unique nonlocal extreme point, when such a point exists.

\begin{theorem}[Strict-gap uniqueness]\label{thm:strict-gap}
Every Bell functional with $Q_{\vB}>L_{\vB}$ has a unique quantum maximizer.
\end{theorem}
\begin{proof}
A strict Bell gap forces frustrated signs by Proposition~\ref{prop:force-frustration}(a). Theorem~\ref{thm:frustrated-face} therefore applies. No local behavior can attain $Q_{\vB}$ when $Q_{\vB}>L_{\vB}$, so the nonempty maximizing face is the singleton $\{\vQ_{\vtheta}\}$.
\end{proof}

\section{From the uniqueness theorem to the decomposition of a general quantum correlation}\label{app:decomposition}
We write $\overline{\vQ}$ for an unrestricted quantum component, to distinguish it from the input $\vQ$.
Stage I identifies a unique behavior once a Bell functional has a strict quantum advantage. The task now is different: given an arbitrary quantum behavior $\vQ$, we must construct a functional to which that result applies. Minimizing the quantum weight will produce it.

The logic is to start with an interior nonlocal $\vQ$, put it on the boundary of the smallest admissible mixture set, and choose a supporting functional there. Interiority of $\vQ$ in $\cQ$ then forces a strict Bell gap. Theorem~\ref{thm:strict-gap} fixes the quantum component, and a final compactness argument reaches the boundary of $\cQ$. This stage uses Proposition~\ref{prop:convex-hull} and Theorem~\ref{thm:strict-gap}, but not the shared-support consequence.

\begin{center}
\begin{minipage}{0.94\textwidth}
\textbf{Proof map for Stage II}\par\smallskip
Fix $\vQ\in\intt\cQ\setminus\cL$, $K_\lambda=\lambda\cQ+(1-\lambda)\cL$. Compactness + nonlocality + interiority $\Rightarrow$ an attained minimum $\lambda_*\in(0,1)$. Minimality $\Rightarrow \vQ\in\partial K_{\lambda_*}\Rightarrow$ a nonzero supporting functional $\beta_{\vB}$ with $\beta_{\vB}(\vQ)=\lambda_*Q_{\vB}+(1-\lambda_*)L_{\vB}$. $\vQ\in\intt\cQ$ and $\vB\ne0\Rightarrow\beta_{\vB}(\vQ)<Q_{\vB}$; together with the support equality $\Rightarrow Q_{\vB}>L_{\vB}$. An attained decomposition: $\vQ=\lambda_*\overline{\vQ}+(1-\lambda_*)\vL_* $; support equality + positive weights $\Rightarrow\beta_{\vB}(\overline{\vQ})=Q_{\vB}$. Stage I $\Rightarrow \overline{\vQ}=\vQ_{\vtheta_*}\in\cS\Rightarrow \vQ=\lambda_*\vQ_{\vtheta_*}+(1-\lambda_*)\vL_*$ for interior nonlocal $\vQ$. Interior approximation + compactness of $[0,1]\times\cS\times\cL$ $\Rightarrow$ the decomposition for every $\vQ\in\cQ$.
\end{minipage}
\end{center}

\subsection{A full-dimensional local reference set}\label{appcontents:C:1}
We first need a local point with a full neighborhood inside $\cQ$. It will allow us to move boundary behaviors into the interior and will also ensure that the mixture sets below are full-dimensional.
\begin{lemma}[The local origin is an interior point]\label{lem:local-origin}
The local polytope $\cL$ is full-dimensional in $\mathbb R^8$, and its origin is an interior point.
\end{lemma}
\begin{proof}
On the sixteen assignments $(\eps_0,\eps_1,\eta_0,\eta_1)\in\{\pm1\}^4$, the constant function and the eight behavior coordinates are the distinct Walsh characters
\begin{equation*}
1,\quad\eps_0,\quad\eps_1,\quad\eta_0,\quad\eta_1,\quad\eps_0\eta_0,\quad\eps_0\eta_1,\quad\eps_1\eta_0,\quad\eps_1\eta_1.
\end{equation*}
They are mutually orthogonal under the uniform average. Hence no nonzero affine functional vanishes on all deterministic vertices, so their affine hull is $\mathbb R^8$.

Their uniform average is the origin. If a supporting hyperplane passed through this strictly positive average of all vertices, linearity would force it to contain every vertex. That would contradict full dimensionality. Thus the origin is interior to $\cL$.
\end{proof}
For $0\le\lambda\le1$, define
\begin{equation}\label{eq:Ktau}
K_\lambda=\lambda\cQ+(1-\lambda)\cL=\{\lambda \overline{\vQ}+(1-\lambda)\vL:\overline{\vQ}\in\cQ,\vL\in\cL\}.
\end{equation}
These are compact convex sets, with $K_0=\cL$ and $K_1=\cQ$. Each contains $\cL$, and hence is full-dimensional. They are nested: if $0\le\lambda_1\le\lambda_2$ with $\lambda_2>0$, then
\begin{equation*}
\lambda_1\overline{\vQ}+(1-\lambda_1)\vL=\lambda_2\left(\frac{\lambda_1}{\lambda_2}\overline{\vQ}+\left(1-\frac{\lambda_1}{\lambda_2}\right)\vL\right)+(1-\lambda_2)\vL\in K_{\lambda_2},
\end{equation*}
because the parenthesized behavior belongs to $\cQ$. The parameter $\lambda$ is the weight assigned to an unrestricted quantum component.

\subsection{An interior nonlocal behavior has a minimum quantum weight}\label{appcontents:C:2}
Fix $\vQ\in\intt\cQ\setminus\cL$, and choose a local reference point $\vL_{\mathrm{ref}}\in\intt\cL$, for example the origin. An arbitrary decomposition of $\vQ$ need not select a Bell maximizer. We therefore minimize its quantum weight:
\begin{equation}\label{eq:tau-min}
\lambda_*:=\min\{\lambda\in[0,1]:\vQ\in K_\lambda\}.
\end{equation}
The minimum is attained. Indeed, the feasible triples $(\lambda,\overline{\vQ},\vL)$ form a nonempty closed subset of the compact product $[0,1]\times\cQ\times\cL$, and their $\lambda$-coordinates form a compact set. Nonlocality gives $\lambda_*>0$. Also $\lambda_*<1$: since $\vQ$ is interior to $\cQ$, there is a small $\eps>0$ such that $\overline{\vQ}'=\vQ+\eps(\vQ-\vL_{\mathrm{ref}})\in\cQ$. Then
\begin{equation*}
\vQ=\frac{1}{1+\eps}\overline{\vQ}'+\frac{\eps}{1+\eps}\vL_{\mathrm{ref}},
\end{equation*}
which uses quantum weight less than one. Thus we may choose an attained decomposition
\begin{equation}\label{eq:tau-attained}
\vQ=\lambda_*\overline{\vQ}+(1-\lambda_*)\vL_*,\qquad \overline{\vQ}\in\cQ,\quad \vL_*\in\cL,\quad0<\lambda_*<1.
\end{equation}
\begin{lemma}[Minimality puts the behavior on the boundary]\label{lem:minimum-boundary}
The point $\vQ$ belongs to $\partial K_{\lambda_*}$.
\end{lemma}
\begin{proof}
Suppose instead that $\vQ\in\intt K_{\lambda_*}$. For some small $\eps>0$, the point $\vQ'=\vQ+\eps(\vQ-\vL_{\mathrm{ref}})$ would also belong to $K_{\lambda_*}$. Write $\vQ'=\lambda_*\overline{\vQ}'+(1-\lambda_*)\vL'$ with $\overline{\vQ}'\in\cQ$ and $\vL'\in\cL$. Solving for $\vQ$ gives
\begin{equation}\label{eq:smaller-weight}
\vQ=\frac{\lambda_*}{1+\eps}\overline{\vQ}'+\frac{1-\lambda_*+\eps}{1+\eps}\left(\frac{(1-\lambda_*)\vL'+\eps \vL_{\mathrm{ref}}}{1-\lambda_*+\eps}\right).
\end{equation}
The parenthesized behavior is local. Hence this is a valid decomposition with quantum weight $\lambda_* /(1+\eps)<\lambda_*$, contradicting minimality.
\end{proof}

\subsection{The supporting functional has a strict quantum gap}\label{appcontents:C:3}
The boundary just obtained belongs to $K_{\lambda_*}$, not initially to $\cQ$. The support function of $K_{\lambda_*}$ connects the two sets. Because $K_{\lambda_*}$ is full-dimensional, the boundary point $\vQ$ admits a nonzero supporting functional $\beta_{\vB}$. Maximizing its two contributions in \eqref{eq:Ktau} independently gives
\begin{equation}\label{eq:tau-support}
\beta_{\vB}(\vQ)=h_{K_{\lambda_*}}(\vB)=\lambda_*Q_{\vB}+(1-\lambda_*)L_{\vB}.
\end{equation}
Since $\vQ\in\intt\cQ$ and $\vB\ne0$, we also have $\beta_{\vB}(\vQ)<Q_{\vB}$. Indeed, a sufficiently small move from $\vQ$ in the direction of $\vB$ stays in $\cQ$ and strictly increases the functional. Subtracting \eqref{eq:tau-support} from $Q_{\vB}$ therefore yields
\begin{equation}\label{eq:strict-gap-step}
0<Q_{\vB}-\beta_{\vB}(\vQ)=(1-\lambda_*)(Q_{\vB}-L_{\vB}).
\end{equation}
As $\lambda_*<1$, it follows that $Q_{\vB}>L_{\vB}$. This is the central geometric step: minimizing the weight produces a supporting functional with exactly the strict gap required by Stage I.

Apply this functional to the attained decomposition \eqref{eq:tau-attained}. We have
\begin{equation*}
\beta_{\vB}(\overline{\vQ})\le Q_{\vB},\qquad\beta_{\vB}(\vL_*)\le L_{\vB}.
\end{equation*}
The weighted sum of these two values equals \eqref{eq:tau-support}. Equivalently,
\begin{equation}\label{eq:tau-saturation}
0=\lambda_*[Q_{\vB}-\beta_{\vB}(\overline{\vQ})]+(1-\lambda_*)[L_{\vB}-\beta_{\vB}(\vL_*)].
\end{equation}
Both weights are positive, and both bracketed terms are nonnegative. Thus both terms vanish: $\overline{\vQ}$ is a quantum maximizer and $\vL_*$ is a local maximizer of the supporting functional.

Theorem~\ref{thm:strict-gap} gives a unique quantum maximizer $\vQ_{\vtheta_*}$, so $\overline{\vQ}=\vQ_{\vtheta_*}$. Moreover, $\vQ_{\vtheta_*}\in\cS$: the maximum over the compact family $\cS$ is attained and equals the maximum over $\cQ=\conv\cS$, and uniqueness identifies that maximizer with $\vQ_{\vtheta_*}$. We have obtained
\begin{equation}\label{eq:interior-decomp}
\vQ=\lambda_*\vQ_{\vtheta_*}+(1-\lambda_*)\vL_*,\qquad \vQ_{\vtheta_*}\in\cS,\quad \vL_*\in\cL.
\end{equation}
The same saturation argument applies to every minimum-weight decomposition of this $\vQ$: each of its quantum components must be the unique maximizer of the chosen supporting functional.

\subsection{Extending the decomposition to the boundary}\label{appcontents:C:4}
The strict-gap argument used interiority, so it cannot be applied directly to an arbitrary boundary behavior. Instead, approximate such a behavior by interior points and pass to a limit.

\begin{proof}[Completion of the proof of Theorem~\ref{thm:decomp}]
For $\vQ\in\cL$, use $\lambda=0$ in \eqref{eq:decomp-main}. Otherwise choose $\vL_{\mathrm{ref}}\in\intt\cL\subseteq\intt\cQ$ and define
\begin{equation}\label{eq:boundary-approx}
\vQ_\eps=(1-\eps)\vQ+\eps \vL_{\mathrm{ref}},\qquad0<\eps<1.
\end{equation}
Each $\vQ_\eps$ is interior to $\cQ$. Indeed, if the Euclidean ball $B(\vL_{\mathrm{ref}},\delta)$ lies in $\cQ$, convexity gives $B(\vQ_\eps,\eps\delta)\subseteq\cQ$.

For every $\eps$, there is a decomposition
\begin{equation}\label{eq:epsilon-decomp}
\vQ_\eps=\lambda_\eps \vQ_{\vtheta_{\eps}}+(1-\lambda_\eps)\vL_\eps,\quad\lambda_\eps\in[0,1],\quad \vQ_{\vtheta_{\eps}}\in\cS,\quad \vL_\eps\in\cL.
\end{equation}
Use \eqref{eq:interior-decomp} when $\vQ_\eps$ is nonlocal, and the trivial local decomposition otherwise.

The product $[0,1]\times\cS\times\cL$ is compact. Along a sequence $\eps\downarrow0$, choose a convergent subsequence of the triples, with limit $(\lambda,\vQ_{\vtheta},\vL)$. Taking the limit in \eqref{eq:epsilon-decomp} gives
\begin{equation*}
\vQ=\lambda \vQ_{\vtheta}+(1-\lambda)\vL,\qquad\lambda\in[0,1],\quad \vQ_{\vtheta}\in\cS,\quad \vL\in\cL.
\end{equation*}
This proves the nontrivial inclusion in \eqref{eq:decomp-main}. The reverse inclusion follows from $\cS,\cL\subseteq\cQ$ and the convexity of $\cQ$.
\end{proof}
The two stages are now joined: minimum weight supplies a strict-gap supporting functional, uniqueness fixes its quantum component, and compactness preserves the one-point-plus-local form at the boundary.

\section{Normal-cone construction}\label{app:normal-cone}
Let $\vQ_{\vtheta}$ be a nonlocal exposed point of $\cQ$. We classify its exposed faces through supporting Bell functionals $\vB=(\alpha_0,\alpha_1,\beta_0,\beta_1,c_{00},c_{01},c_{10},c_{11})$, where $\alpha_x$ and $\beta_y$ multiply $A_x$ and $B_y$. Their normals belong to
\begin{equation}\label{eq:normal-D}
\scrN(\vQ_{\vtheta})=\{\vB\in\mathbb R^8:\vB\cdot(\vQ-\vQ_{\vtheta})\le0\ \text{for every }\vQ\in\cQ\}.
\end{equation}
Since the origin lies in $\intt\cQ$, every nonzero $\vB\in\scrN(\vQ_{\vtheta})$ satisfies $\vB\cdot \vQ_{\vtheta}>0$. Thus we may normalize it by $\vB\cdot \vQ_{\vtheta}=1$. For an interior, differentiable five-parameter realization ($0<\theta<\pi/2$), a normalized supporting functional must satisfy
\begin{equation}\label{eq:tangent-normalization}
\underbrace{\vB\cdot \vQ_{\vtheta}=1}_{\text{normalization}},\qquad\underbrace{T_{\vtheta}^{T}\vB=0}_{\text{tangency}},
\end{equation}
where $T_{\vtheta}=\partial\vQ_{\vtheta}/\partial(\theta,a_0,a_1,b_0,b_1)$ is the $8\times5$ tangent matrix. These linear conditions are necessary but not sufficient: the candidate is a genuine supporting normal precisely when
\begin{equation}\label{eq:certify-normal}
\max_{\vQ\in\cQ}\vB\cdot \vQ=\max_{\vtheta'}\vB\cdot\vQ_{\vtheta'}=1.
\end{equation}
Here the second equality uses $\cQ=\conv\cS$.

\subsection{Rank of the tangent matrix}\label{appcontents:D:1}
\begin{lemma}[Full column rank]\label{lem:tangent-rank}
Consider the parametrization
\begin{equation}\label{eq:param-D}
A_x=\cos\theta\cos a_x,\qquad B_y=\cos\theta\cos b_y,\qquad E_{xy}=\cos a_x\cos b_y+\sin\theta\sin a_x\sin b_y.
\end{equation}
If $0<\theta<\pi/2$ and both local measurement pairs are noncommuting,
\begin{equation}\label{eq:noncommuting-D}
\sin(a_0-a_1)\ne0,\qquad\sin(b_0-b_1)\ne0,
\end{equation}
then $\rank T_{\vtheta}=5$.
\end{lemma}
\begin{proof}
Put $s=\sin\theta$ and $c=\cos\theta$, with $s,c>0$. Differentiating Eq.~\eqref{eq:param-D}, we obtain
\begin{equation}\label{eq:tangent-matrix}
T_{\vtheta}=\begin{pmatrix}
-s\cos a_0&-c\sin a_0&0&0&0\\
-s\cos a_1&0&-c\sin a_1&0&0\\
-s\cos b_0&0&0&-c\sin b_0&0\\
-s\cos b_1&0&0&0&-c\sin b_1\\
c\sin a_0\sin b_0&\tau^{\hat A}_{00}&0&\tau^{\hat B}_{00}&0\\
c\sin a_0\sin b_1&\tau^{\hat A}_{01}&0&0&\tau^{\hat B}_{01}\\
c\sin a_1\sin b_0&0&\tau^{\hat A}_{10}&\tau^{\hat B}_{10}&0\\
c\sin a_1\sin b_1&0&\tau^{\hat A}_{11}&0&\tau^{\hat B}_{11}
\end{pmatrix},
\end{equation}
where
\begin{equation}\label{eq:tau-D}
\tau^{\hat A}_{xy}=-\sin a_x\cos b_y+s\cos a_x\sin b_y,\qquad\tau^{\hat B}_{xy}=-\cos a_x\sin b_y+s\sin a_x\cos b_y.
\end{equation}
The four marginal derivatives give a nonzero diagonal block in the four angular columns whenever all angle sines are nonzero. Eliminating these columns from the correlator rows leaves the four entries
\begin{equation}\label{eq:residual-D}
R_{xy}=\frac{(\sin a_x\sin b_y+s\cos a_x\cos b_y)^2-s^2}{c\sin a_x\sin b_y}.
\end{equation}
Thus $\rank T_{\vtheta}=5$ if at least one $R_{xy}$ is nonzero.

Suppose all four residuals vanish. By output relabeling, which preserves the rank, we may take $0<a_x,b_y<\pi$. Since $0<s<1$ and all sines are then positive,
\begin{equation}\label{eq:residual-ineq}
\sin a_x\sin b_y+s\cos a_x\cos b_y>-s.
\end{equation}
The equations $R_{xy}=0$ therefore imply
\begin{equation}\label{eq:residual-zero}
\sin a_x\sin b_y+s\cos a_x\cos b_y=s\qquad(x,y\in\{0,1\}).
\end{equation}
Subtracting the cases $x=0,1$ gives
\begin{equation}\label{eq:linear-D}
\begin{pmatrix}\sin b_0&s\cos b_0\\\sin b_1&s\cos b_1\end{pmatrix}\begin{pmatrix}\sin a_0-\sin a_1\\\cos a_0-\cos a_1\end{pmatrix}=0.
\end{equation}
The determinant of the coefficient matrix is $s\sin(b_0-b_1)\ne0$. Therefore $a_0\equiv a_1\pmod {2\pi}$, contrary to Alice's noncommutativity. Hence some $R_{xy}\ne0$, and the rank is five in this case.

It remains to consider vanishing angle sines. By noncommutativity, at most one vanishes on each side. Since $T_{\vtheta}$ has five columns, one nonzero $5\times5$ minor suffices. Up to input relabeling, three cases remain.

(i) $\sin a_0=0$ and $\sin a_1\sin b_0\sin b_1\ne0$. Selecting the rows $(A_0,A_1,B_0,B_1,E_{00})$ yields
\begin{equation}\label{eq:minor-a}
M_{\hat A}=\begin{pmatrix}
-s\cos a_0&0&0&0&0\\
-s\cos a_1&0&-c\sin a_1&0&0\\
-s\cos b_0&0&0&-c\sin b_0&0\\
-s\cos b_1&0&0&0&-c\sin b_1\\
0&s\cos a_0\sin b_0&0&-\cos a_0\sin b_0&0
\end{pmatrix}.
\end{equation}
Since $\cos^2a_0=1$,
\begin{equation}\label{eq:minor-a-det}
|\det M_{\hat A}|=s^2c^3|\sin a_1|\sin^2b_0|\sin b_1|>0.
\end{equation}
(ii) $\sin b_0=0$ and $\sin a_0\sin a_1\sin b_1\ne0$. Selecting $(B_0,A_0,A_1,B_1,E_{00})$ yields
\begin{equation}\label{eq:minor-b}
M_{\hat B}=\begin{pmatrix}
-s\cos b_0&0&0&0&0\\
-s\cos a_0&-c\sin a_0&0&0&0\\
-s\cos a_1&0&-c\sin a_1&0&0\\
-s\cos b_1&0&0&0&-c\sin b_1\\
0&-\sin a_0\cos b_0&0&s\sin a_0\cos b_0&0
\end{pmatrix}.
\end{equation}
Since $\cos^2b_0=1$,
\begin{equation}\label{eq:minor-b-det}
|\det M_{\hat B}|=s^2c^3\sin^2a_0|\sin a_1\sin b_1|>0.
\end{equation}
(iii) $\sin a_0=\sin b_0=0$. Here $\sin a_1\sin b_1\ne0$ by noncommutativity. Selecting $(A_0,A_1,B_1,E_{01},E_{10})$ yields
\begin{equation}\label{eq:minor-ab}
M_{\hat A\hat B}=\begin{pmatrix}
-s\cos a_0&0&0&0&0\\
-s\cos a_1&0&-c\sin a_1&0&0\\
-s\cos b_1&0&0&0&-c\sin b_1\\
0&s\cos a_0\sin b_1&0&0&-\cos a_0\sin b_1\\
0&0&-\sin a_1\cos b_0&s\sin a_1\cos b_0&0
\end{pmatrix}.
\end{equation}
Since $\cos^2a_0=\cos^2b_0=1$,
\begin{equation}\label{eq:minor-ab-det}
|\det M_{\hat A\hat B}|=s^3c^2\sin^2a_1\sin^2b_1>0.
\end{equation}
Thus the tangent matrix has full column rank in every case.
\end{proof}

\subsection{Reduction to two parameters}\label{appcontents:D:2}
For a nonlocal extreme point with $0<\theta<\pi/2$, Lemma~\ref{lem:tangent-rank} provides five independent tangency constraints. Together with the independent normalization $\vB\cdot \vQ_{\vtheta}=1$, they define a two-dimensional affine family of candidate Bell vectors,
\begin{equation}\label{eq:affine-family}
\vB(c_1,c_2)=\vB^{(0)}+c_1\vB^{(1)}+c_2\vB^{(2)}.
\end{equation}
This dimension count concerns the linear candidate space; the normalized normal cone is its subset satisfying the global quantum bound.

\paragraph*{Maximally entangled case.} At $\theta=\pi/2$ the marginals vanish, and the positive-sign parametrization in Eq.~\eqref{eq:param-D} gives $E_{xy}=\cos(a_x-b_y)$. It is invariant under
\begin{equation}\label{eq:common-shift}
a_x\mapsto a_x+t,\qquad b_y\mapsto b_y+t.
\end{equation}
At regular nonlocal points the tangent rank is four, because the simultaneous angular shift is redundant. Nevertheless, for every shifted realization of a quantum maximizer, one-sided optimality at $\theta=\pi/2$ implies
\begin{equation}\label{eq:g-shift}
g(t):=\sum_{x=0}^{1}\alpha_x\cos(a_x+t)+\sum_{y=0}^{1}\beta_y\cos(b_y+t)\le0.
\end{equation}
Since $g(t+\pi)=-g(t)$, one has $g\equiv0$, hence
\begin{equation}\label{eq:shift-constraints}
\begin{gathered}
\sum_x\alpha_x\cos a_x+\sum_y\beta_y\cos b_y=0,\\
\sum_x\alpha_x\sin a_x+\sum_y\beta_y\sin b_y=0.
\end{gathered}
\end{equation}
The second equation supplies the additional linear constraint needed to recover a two-dimensional candidate family at a regular maximally entangled nonlocal point.

\paragraph*{Finite face enumeration.}
For a fixed nonlocal exposed point $\vQ_{\vtheta}$, the quantum supporting condition is
\begin{equation}\label{eq:A-qtheta}
\vB\cdot \vQ_{\vtheta}=\max_{\vtheta'}\vB\cdot\vQ_{\vtheta'}=1.
\end{equation}
For each subset $I\subseteq\{1,\ldots,16\}$, impose the local equalities and strict inequalities in Theorem~\ref{thm:finite-cone}. These conditions identify exactly the classical vertices in the exposed face. In the regular two-parameter case, the equations
\begin{equation}\label{eq:local-candidates}
\vB(c_1,c_2)\cdot\vL_i=\vB(c_1,c_2)\cdot\vL_j=1
\end{equation}
produce candidate normals when two independent local constraints are active; they need not exhaust the feasible sets $I$.

\begin{proof}[Proof of Theorem~\ref{thm:finite-cone}]
Every nonzero supporting $\vB$ at $\vQ_{\vtheta}$ has frustrated signs: Proposition~\ref{prop:force-frustration}(a) applies when $Q_{\vB}>L_{\vB}$, and Proposition~\ref{prop:force-frustration}(b) applies otherwise, since a nonlocal exposed point has strictly positive joint probabilities \cite{BarizienBancal2025}. Theorem~\ref{thm:frustrated-face} then gives
$F_{\vB}=\conv(\{\vQ_{\vtheta}\}\cup(F_{\vB}\cap\cL))$.
The latter intersection is either empty or a face of the local polytope, generated by its saturated deterministic vertices. This proves Eq.~\eqref{eq:finite-cone}. Because there are only sixteen such vertices, only finitely many exposed faces contain $\vQ_{\vtheta}$.

The displayed feasibility conditions select exactly the vertices indexed by $I$; conversely, every exposed face containing $\vQ_{\vtheta}$ supplies such a functional after normalization. For fixed $\vQ_{\vtheta}$, inclusion of these faces is equivalent to inclusion of their index sets, since the $\vL_i$ are extreme points of $\cQ$. Finally, every maximal proper face is exposed by a supporting functional (Proposition~\ref{prop:extreme-normals}). Hence the inclusion-maximal feasible sets give precisely all maximal proper faces containing $\vQ_{\vtheta}$.
\end{proof}
The finite classification concerns faces, not the number of extreme rays of $\scrN(\vQ_{\vtheta})$. Global quantum feasibility remains essential, while pairwise local saturation is only a computational simplification.

Certain specified patterns of vanishing probabilities force a quantum correlation to be local \cite{Goh2018}, but vanishing probabilities alone do not imply locality, as demonstrated by Hardy-type nonlocal correlations \cite{Hardy1993}. Such zero constraints locate correlations on positivity faces or their intersections; determining whether a particular pattern forces locality requires a separate argument. They do not by themselves establish a complete classification of quantum faces.


\subsection{Certified example: four maximal triangular faces}\label{sec:four-face-example}
Consider the five angles
\begin{equation}\label{eq:face-example-angles}
\vtheta=\left(\frac{11\pi}{6},\frac\pi2,\frac{5\pi}{3},
\frac{2\pi}{5},\frac{9\pi}{20}\right).
\end{equation}
In the coordinate order $(A_0,A_1,B_0,B_1,E_{00},E_{01},E_{10},E_{11})$, the corresponding behavior is
\begin{equation}\label{eq:face-example-point}
\begin{split}
\vQ_{\vtheta}\simeq(&0.135476221,\,0,\,0.078217233,\,0.048340908,\\
 &0.860694299,\,-0.202057149,\,-0.855363194,\,0.939347432).
\end{split}
\end{equation}
Here and below, displayed decimal coordinates are approximations; the exact behavior is defined by Eq.~\eqref{eq:family-main} at the angles in Eq.~\eqref{eq:face-example-angles}.
For $u,v\in\{\pm1\}$, write
\begin{equation}\label{eq:face-example-local}
\vL_{uv}=(u,v,u,v,1,uv,uv,1).
\end{equation}
These four deterministic vertices all saturate the CHSH inequality
$E_{00}+E_{01}-E_{10}+E_{11}\le2$, which is violated by $\vQ_{\vtheta}$.

\paragraph*{Exhausting the local candidates.}
The tangent matrix $T_{\vtheta}$ has rank five, and normalization is independent of its five tangency constraints. Hence
\[
\mathcal T_{\vtheta}:=
\{\vB\in\mathbb R^8:T_{\vtheta}^{T}\vB=0,\ \vB\cdot\vQ_{\vtheta}=1\}
\]
is a two-dimensional affine plane. Let
\[
\mathcal P_{\vtheta}:=
\{\vB\in\mathcal T_{\vtheta}:\vB\cdot\vL\le1
\text{ for every }\vL\in\Ext(\cL)\}.
\]
This set is bounded because the origin is an interior point of $\cL$.
We test all $\binom{16}{2}=120$ pairs of deterministic vertices, not only pairs involving the four vertices in Eq.~\eqref{eq:face-example-local}.
Each pair gives a nonsingular eight-equation system. Exactly four solutions satisfy every local inequality; the other 116 violate at least one. The four feasible pairs are
\begin{equation}\label{eq:face-example-active-pairs}
\begin{aligned}
&k=1:\ (\vL_{++},\vL_{+-}),\qquad
 k=2:\ (\vL_{++},\vL_{-+}),\\
&k=3:\ (\vL_{+-},\vL_{--}),\qquad
 k=4:\ (\vL_{-+},\vL_{--}).
\end{aligned}
\end{equation}
For each pair $(i,j)$, define $\vB^{(k)}$ as the \emph{exact} solution of
\begin{equation}\label{eq:face-example-B-def}
\vB^{(k)}\cdot\vQ_{\vtheta}=1,\qquad
T_{\vtheta}^{T}\vB^{(k)}=0,\qquad
\vB^{(k)}\cdot\vL_i=\vB^{(k)}\cdot\vL_j=1.
\end{equation}
In the coefficient order
$(\alpha_0,\alpha_1,\beta_0,\beta_1,c_{00},c_{01},c_{10},c_{11})$, these solutions are
\begin{align*}
\vB^{(1)}\simeq{}&
(0.014765939568,\,0.159989003728,\,0.207277270078,\,0.046257971966,\\
&\hspace{5mm}0.289599682625,\,0.145629903950,\,-0.351876879644,\,0.488357107728),\\
\vB^{(2)}\simeq{}&
(-0.049515158486,\,0.048283413796,\,0.249242244543,\,0.165025140555,\\
&\hspace{5mm}0.264511183225,\,0.138104444391,\,-0.337831530447,\,0.522180262424),\\
\vB^{(3)}\simeq{}&
(0.156429614678,\,-0.494154919194,\,-0.246726705457,\,0.386991017253,\\
&\hspace{5mm}0.284030578998,\,0.131359691677,\,-0.221656782456,\,0.608805519060),\\
\vB^{(4)}\simeq{}&
(0.082582532173,\,-0.588326035045,\,-0.184800736352,\,0.500574849943,\\
&\hspace{5mm}0.257517527215,\,0.123884957556,\,-0.211636142658,\,0.640264268606).
\end{align*}
Each solution saturates precisely its stated pair. The minimum slack
$1-\vB^{(k)}\cdot\vL$ over the other fourteen vertices is, respectively,
\[
0.031592467906,\quad 0.027162936590,\quad
0.026216167612,\quad 0.028934038155
\]
to the displayed precision. Thus there is no third saturated local vertex.
Since every vertex of a bounded two-dimensional local-constraint polytope has two independent active constraints, the exhaustive pair test gives
\[
\mathcal P_{\vtheta}
=\conv\{\vB^{(1)},\vB^{(2)},\vB^{(3)},\vB^{(4)}\}.
\]
This is a quadrilateral; local feasibility alone does not yet identify it with the quantum normal-cone section.

\paragraph*{Global quantum certification.}
Set $s_0=\cos(a_0-a_1)=-1/2$ and
$t_0=\cos(b_0-b_1)=\cos(19\pi/15)$. For each $k$, define
\[
M_k(s,t)=I-\widehat B^{(k)}(s,t),\qquad
f_k(s,t)=\det M_k(s,t)=a_k(t)s^2+b_k(t)s+c_k(t).
\]
The coefficients are computed from the exact linear systems, using the trace identities of Appendix~\ref{app:faces}.
The five stationarity equations in Eq.~\eqref{eq:face-example-B-def} include variations of the Schmidt angle and collective local measurement rotations. Since $0<\theta<\pi/2$, the corresponding real-state variations span the three-dimensional tangent space orthogonal to $|\psi_\theta\rangle$. Normalization and stationarity therefore make $|\psi_\theta\rangle$ an eigenvector of the Bell operator with eigenvalue one. Measurement stationarity also gives
\[
f_k(s_0,t_0)=\partial_s f_k(s_0,t_0)
=\partial_t f_k(s_0,t_0)=0.
\]
Consequently,
\begin{equation}\label{eq:face-example-discriminant}
4a_k(t)c_k(t)-b_k(t)^2=(t-t_0)^2q_k(t),
\end{equation}
where $q_k$ is a quadratic polynomial.

The two saturated deterministic strategies give endpoint zeros at
$(s,t)=(-1,-1)$ and $(1,1)$. Accordingly,
\[
f_k(-1,t)=(1+t)\ell_{k,-}(t),\qquad
f_k(1,t)=(1-t)\ell_{k,+}(t),
\]
where both $\ell_{k,\pm}$ are affine and satisfy
$\ell_{k,\pm}(t)>0.005$ throughout $[-1,1]$.
To check interior minima in $s$, it suffices to consider $a_k(t)>0$ and
$|b_k(t)|<2a_k(t)$; otherwise an endpoint gives the minimum.
For $k=1,2,3,4$, respectively, the $t$ values satisfying these conditions are contained in
\[
[-0.963,0.557],\quad[-0.957,0.642],\quad
[-0.893,0.505],\quad[-0.9833,0.40030].
\]
On each interval, $q_k$ is concave and its endpoint values are greater than $9\times10^{-6}$. Equation~\eqref{eq:face-example-discriminant} then makes every interior minimum nonnegative. Together with the endpoint checks, this proves $f_k(s,t)\ge0$ on the entire overlap square.

For completeness, a nonnegative determinant alone is not a semidefiniteness certificate. The remaining elementary symmetric coefficients of $M_k$ are
\[
e_1=4,\qquad e_2=6-2w_k,\qquad
 e_3=4-4w_k-8d_k,\qquad e_4=f_k,
\]
where $w_k,d_k$ are the Gram-matrix quantities of Appendix~\ref{app:faces}. Both $e_2$ and $e_3$ are bilinear in $(s,t)$; their corner values certify
$e_2>2.21$ and $e_3>0.06$ throughout $[-1,1]^2$.
All coefficients of
$\det(xI+M_k)=x^4+e_1x^3+e_2x^2+e_3x+e_4$
are therefore nonnegative. A negative eigenvalue of the real symmetric matrix $M_k$ would give a positive root of this polynomial, which is impossible. Thus
\[
M_k(s,t)\succeq0\quad\text{for every }(s,t)\in[-1,1]^2,
\qquad Q_{\vB^{(k)}}=L_{\vB^{(k)}}=1.
\]
All strict numerical inequalities above, and the classification of the 120 candidate pairs, were checked with 75-digit outward-rounded interval arithmetic. The exact saturation and factorization identities follow from the defining equations, rather than from rounding a numerical residual to zero.

\paragraph*{The complete normal cone and its maximal faces.}
Every nonzero normal at $\vQ_{\vtheta}$ can be normalized to have value one there. Its normalized section satisfies
\[
\mathcal A_{\vtheta}:=
\{\vB\in\scrN(\vQ_{\vtheta}):\vB\cdot\vQ_{\vtheta}=1\}
\subseteq\mathcal P_{\vtheta}.
\]
Conversely, all four vertices of $\mathcal P_{\vtheta}$ have just been certified as quantum supporting functionals. Convexity gives
$\mathcal P_{\vtheta}\subseteq\mathcal A_{\vtheta}$, and therefore the normalized section is exactly the quadrilateral. In particular,
\begin{equation}\label{eq:face-example-normal-cone}
\scrN(\vQ_{\vtheta})=
\cone\{\vB^{(1)},\vB^{(2)},\vB^{(3)},\vB^{(4)}\}.
\end{equation}
These four vectors generate the four extreme rays of the cone.
Their mean $\bar{\vB}=\tfrac14\sum_{k=1}^4\vB^{(k)}$ satisfies
\[
Q_{\bar{\vB}}=1,\qquad
L_{\bar{\vB}}\simeq0.985975756313881<1.
\]
Strict-gap uniqueness therefore confirms that $\vQ_{\vtheta}$ is a nonlocal exposed point.
Using the maximizing-face structure and the exact local saturation sets, we obtain
\begin{equation}\label{eq:face-example-four-faces}
\begin{aligned}
F_1&=\conv\{\vQ_{\vtheta},\vL_{++},\vL_{+-}\},\\
F_2&=\conv\{\vQ_{\vtheta},\vL_{++},\vL_{-+}\},\\
F_3&=\conv\{\vQ_{\vtheta},\vL_{+-},\vL_{--}\},\\
F_4&=\conv\{\vQ_{\vtheta},\vL_{-+},\vL_{--}\}.
\end{aligned}
\end{equation}
Each is a two-dimensional triangle: the two local vertices are distinct and their affine line has CHSH value two, whereas $\vQ_{\vtheta}$ has a larger value.
Each face is maximal. Indeed, a larger proper face containing $F_k$ would lie in a supporting hyperplane whose normalized functional satisfies the same eight equations as $\vB^{(k)}$, including saturation at the same two local vertices. Uniqueness of that linear-system solution forces the same functional and the same exposed face, a contradiction.
Conversely, every normalized normal is a convex combination of the four $\vB^{(k)}$, and its exposed face is the intersection of the faces with positive coefficients. Hence the four triangles in Eq.~\eqref{eq:face-example-four-faces} are all the maximal proper faces containing $\vQ_{\vtheta}$.

\section{Weight-free algorithms for quantum decompositions}\label{app:weight-free}

The decomposition theorem is an existence statement. We now convert it into a membership test without unknown mixture weights. The starting point is the five-parameter behavior $\vQ_{\vtheta}$ in Eq.~\eqref{eq:family-app}. Throughout, the input $\vQ$ is assumed to be normalized, no-signaling, and to have nonnegative conditional probabilities. Inputs satisfying all CHSH inequalities are local and hence quantum \cite{Fine1982}; we focus on $S(\vQ)>2$. We first explain the sparse two-vertex elimination and then give a complete eight-vertex formulation that applies to every nonlocal quantum behavior.

\subsection{Method 1: Elimination for a fixed pair of local vertices}\label{sec:weight-free-pair}

Choose two distinct deterministic behaviors $\vL_i,\vL_j$ saturating a CHSH inequality. A useful sparse ansatz is
\begin{equation}\label{eq:F-pair-mixture}
\begin{gathered}
\vQ=\lambda_0\vQ_{\vtheta}+\lambda_i\vL_i+\lambda_j\vL_j,\\
\lambda_0,\lambda_i,\lambda_j\ge0,\qquad
\lambda_0+\lambda_i+\lambda_j\le1.
\end{gathered}
\end{equation}
The remaining coefficient multiplies the zero behavior $\vQ_0=\vec0$, for which $p_{\vQ_0}(ab|xy)=1/4$. Thus the right-hand side represents a normalized convex mixture even when the displayed weights sum to less than one. If $\vQ$ is nonlocal, necessarily $\lambda_0>0$; otherwise it would be a mixture of local behaviors. Dividing by this positive coefficient and setting
\begin{equation*}
z_0=\frac1{\lambda_0},\qquad
z_i=\frac{\lambda_i}{\lambda_0},\qquad
z_j=\frac{\lambda_j}{\lambda_0},
\end{equation*}
gives
\begin{equation*}
\vQ_{\vtheta}=z_0\vQ-z_i\vL_i-z_j\vL_j.
\end{equation*}
The three $z$-coordinates are useful intermediate quantities, but they are not additional nonlinear unknowns: once $\vtheta$ is fixed, they will be determined by a linear system.

\paragraph*{Why the elimination has exactly five equations.}
Define the fixed $8\times3$ matrix
\begin{equation*}
C_{ij}=[\,\vQ,-\vL_i,-\vL_j\,].
\end{equation*}
For a valid input strictly violating the same CHSH facet as the selected pair, $C_{ij}$ has column rank three.  Pairs taken from unrelated facets can instead be checked by a rank-revealing factorization.

Select $D_{ij}\in\mathbb R^{5\times8}$ and $K_{ij}\in\mathbb R^{3\times8}$ satisfying
\begin{equation}\label{eq:F-pair-projectors}
D_{ij}C_{ij}=0,\quad \operatorname{rank}D_{ij}=5,
\quad K_{ij}C_{ij}=I_3.
\end{equation}
The five rows of $D_{ij}$ span the orthogonal complement of the three-dimensional column space of $C_{ij}$; hence
\begin{equation*}
D_{ij}\vQ_{\vtheta}=0
\iff
\vQ_{\vtheta}\in\operatorname{Col}(C_{ij}).
\end{equation*}
This gives five scalar nonlinear equations in the five angles. Once they hold, $K_{ij}$, a left inverse of $C_{ij}$, extracts the unique linear coefficients:
\begin{equation}\label{eq:F-z-coordinates}
\begin{pmatrix}z_0(\vtheta)\\ z_i(\vtheta)\\ z_j(\vtheta)\end{pmatrix}
=K_{ij}\vQ_{\vtheta}.
\end{equation}
Because $C_{ij}K_{ij}$ acts as the identity on $\operatorname{Col}(C_{ij})$, the projected equalities imply the original vector identity
\begin{equation*}
\vQ_{\vtheta}=z_0\vQ-z_i\vL_i-z_j\vL_j.
\end{equation*}

\paragraph*{Recovering and checking the weights.}
The physical requirements are $z_i,z_j\ge0$ and
\begin{equation*}
\lambda_0+\lambda_i+\lambda_j
=\frac{1+z_i+z_j}{z_0}\le1.
\end{equation*}
For $z_0>0$ this is equivalent to $z_0-z_i-z_j\ge1$. Conversely, these three inequalities already imply $z_0\ge1>0$, so there is no need to impose a separate denominator condition. We have proved the equivalence between Eq.~\eqref{eq:F-pair-mixture} and
\begin{equation}\label{eq:F-pair-test}
\begin{gathered}
D_{ij}\vQ_{\vtheta}=0,\qquad
z_i(\vtheta),z_j(\vtheta)\ge0,\\
z_0(\vtheta)-z_i(\vtheta)-z_j(\vtheta)\ge1.
\end{gathered}
\end{equation}
The weights are obtained without another nonlinear solve:
\begin{equation}\label{eq:F-pair-recovery}
\lambda_0=\frac1{z_0},\qquad
\lambda_i=\frac{z_i}{z_0},\qquad
\lambda_j=\frac{z_j}{z_0}.
\end{equation}
The residual weight on $\vQ_0$ equals $(z_0-1-z_i-z_j)/z_0\ge0$. Conversely, every valid pairwise decomposition produces precisely these $z$-coordinates and satisfies Eq.~\eqref{eq:F-pair-test}.

\paragraph*{Explicit construction.}
For numerical calculations, take a full QR decomposition $C_{ij}=U_1H$, where $U=(U_1,U_2)$ is orthogonal, $U_1\in\mathbb R^{8\times3}$, $U_2\in\mathbb R^{8\times5}$ and $H\in\mathbb R^{3\times3}$ is invertible. Then
\begin{equation*}
D_{ij}=U_2^T,\qquad K_{ij}=H^{-1}U_1^T.
\end{equation*}
For exact rational inputs, one can avoid floating-point QR. Choose three row indices $\mathsf p$ for which the submatrix $C_{\mathsf p}$ is invertible and let $\mathsf r$ consist of the remaining five rows. The equations and coefficient recovery become
\begin{equation*}
\begin{gathered}
\vQ_{\vtheta,\mathsf r}
-C_{\mathsf r}C_{\mathsf p}^{-1}
\vQ_{\vtheta,\mathsf p}=0,\\
(z_0,z_i,z_j)^T=C_{\mathsf p}^{-1}\vQ_{\vtheta,\mathsf p}.
\end{gathered}
\end{equation*}
All five displayed residual equations are required, together with the three weight inequalities. The pairwise criterion is exact for the \emph{chosen} pair but does not cover arbitrary local mixtures; even failure for all $\binom82=28$ pairs cannot exclude decompositions involving three or more local vertices. For a prescribed pair, thin QR factorization of the $8\times3$ matrix $C_{ij}$ costs $O(8\cdot3^2)$ arithmetic operations; completing the orthogonal basis is also a fixed-size operation. 

\paragraph*{Exact semialgebraic complexity.}
Introduce $u_x=\cos a_x$, $v_x=\sin a_x$, $w_y=\cos b_y$, $r_y=\sin b_y$, and $c=\cos\theta$, $s=\sin\theta$. The ten real variables obey
\[
u_x^2+v_x^2=1,\quad w_y^2+r_y^2=1\quad(x,y=0,1),
\qquad c^2+s^2=1,\qquad c,s\ge0.
\]
The quantum coordinates are
\[
\widetilde A_x=cu_x,\qquad
\widetilde B_y=cw_y,\qquad
\widetilde E_{xy}=u_xw_y+sv_xr_y,
\]
and have degree at most three. For the complete criterion, the input probabilities and $\delta$ are constants; hence $g_i$ also has degree at most three, rather than six. The complete system has five quadratic equalities and ten inequalities (the eight $g_i\ge0$ and $c,s\ge0$). The pairwise system has ten equalities (five residual equations and five circle equations) and five inequalities (three weight conditions and $c,s\ge0$). Neither formulation introduces unknown mixture weights.

For rational inputs with total encoding length $L$, standard existential real-algebraic algorithms admit bounds of the form $(md)^{O(n)}\operatorname{poly}(L)$, where $n,m,d$ count variables, constraints, and maximum degree \cite{BasuPollackRoy2006}. Here $n=10$, $m=15$, and $d\le3$ are fixed, giving a polynomial input-bit-length bound for exact feasibility in this Bell scenario. The constants can be large, and this bound neither guarantees fast local optimization nor specifies a wall-clock time. Rational pivot elimination preserves exact coefficients in the pairwise formulation; unverified floating-point QR factors cannot be used as exact input data. Experimental uncertainty must be specified separately from exact rational membership.

\paragraph*{Recorded single-input benchmark.}

\subsection{Direct elimination for a fixed post-quantum input}\label{appcontents:E:1}
We illustrate the decomposition search without assuming a separating Bell functional. Fix
\begin{equation}\label{eq:example-def}
\eps=10^{-4},\qquad k=\frac{1+\eps}{3}=\frac{10001}{30000},\qquad r=\frac1{\sqrt2},
\end{equation}
and consider the input
\begin{equation}\label{eq:example-point}
\vQ_*=k(0,0,-2,0,r,2+r,r,-2-r).
\end{equation}

The coordinate order is $(A_0,A_1,B_0,B_1,E_{00},E_{01},E_{10},E_{11})$. For a candidate pair of deterministic vertices, we solve
\begin{equation}\label{eq:example-mixture}
\vQ_*=\lambda_0\vQ_{\vtheta}+\lambda_1\vL_1+\lambda_2\vL_2,\qquad\lambda_0,\lambda_1,\lambda_2\ge0,\quad\lambda_0+\lambda_1+\lambda_2\le1.
\end{equation}
Here $\vQ_{\vtheta}$ has the five-parameter realization
\begin{equation}\label{eq:example-parameter}
A_x=c\cos a_x,\quad B_y=c\cos b_y,\quad E_{xy}=\cos a_x\cos b_y+s\sin a_x\sin b_y,\quad s=\sin\theta,\quad c=\cos\theta.
\end{equation}
In the following, the symbols $A_x,B_y,E_{xy}$ refer to the unknown component $\vQ_{\vtheta}$, not to the fixed input $\vQ_*$. 

\paragraph*{A pair admitting a complete analytic calculation.} Take
\begin{equation}\label{eq:two-local}
\vL_1=(-1,1,-1,-1,1,1,-1,-1),\qquad \vL_2=(1,-1,-1,1,-1,1,1,-1).
\end{equation}
Under the deterministic labeling $A_x=(-1)^{ax+b}$, $B_y=(-1)^{cy+d}$, these are $\vL_{1101}$ and $\vL_{1011}$. This pair is the eighth in the fixed pseudorandom ordering used in our accompanying numerical illustration (seed 20261022).

The two local weights are eliminated using the $A_0$ and $B_0$ coordinates:
\begin{equation}\label{eq:example-weights}
\lambda_1=k+\frac{\lambda_0}{2}(B_0+A_0),\qquad \lambda_2=k+\frac{\lambda_0}{2}(B_0-A_0).
\end{equation}
The remaining consistency conditions are
\begin{equation}\label{eq:example-five-eqs}
\begin{gathered}
A_0+A_1=0,\quad B_1-A_0=0,\quad E_{01}+E_{11}=0,\\
E_{10}-E_{00}-2A_0=0,\quad E_{00}+A_0-E_{01}-B_0=0.
\end{gathered}
\end{equation}
The quantum weight is recovered from
\begin{equation}\label{eq:example-mu0}
\lambda_0=\frac{kr}{E_{00}+A_0}.
\end{equation}
The denominator must be nonzero; for nonnegative weights it is positive. Indeed, $E_{00}+E_{10}=0$ at both local vertices, whereas the target has this sum equal to $2kr>0$, so $\lambda_0=0$ is impossible. Equations~\eqref{eq:example-five-eqs} are five trigonometric consistency equations in the five angles, with all mixture weights eliminated.

\paragraph*{Excluding the nonmaximally entangled branch.} Suppose $0\le\theta<\pi/2$, so $c>0$. The first two equations imply
\begin{equation}\label{eq:cos-relations}
\cos a_1=-\cos a_0,\qquad\cos b_1=\cos a_0.
\end{equation}
The third then becomes
\begin{equation}\label{eq:sin-relation}
s(\sin a_0+\sin a_1)\sin b_1=0.
\end{equation}
If $s=0$, the parametrization immediately gives $E_{00}+E_{10}=0$. If $s>0$ and $\sin b_1\ne0$, then $\sin a_1=-\sin a_0$, with the same consequence. If $s>0$ and $\sin b_1=0$, then $\cos a_0=\pm1$ and $\sin a_0=\sin a_1=0$, again giving $E_{00}+E_{10}=0$. All cases contradict
\begin{equation}\label{eq:example-positive}
\lambda_0(E_{00}+E_{10})=2kr>0.
\end{equation}
Thus no solution exists in this branch.

\paragraph*{Solving the maximally entangled branch.} At $\theta=\pi/2$, all marginals of $\vQ_{\vtheta}$ vanish. Equation~\eqref{eq:example-weights} gives $\lambda_1=\lambda_2=k$, and the correlators must have the form
\begin{equation}\label{eq:example-correlators}
(E_{00},E_{01},E_{10},E_{11})=(t,t,t,-t),\qquad t=\frac{kr}{\lambda_0}>0.
\end{equation}
These correlators are inner products of planar unit vectors, since $E_{xy}=\cos(a_x-b_y)$. Put $z=\cos(b_0-b_1)$. The cases $z=\pm1$ are incompatible with $t>0$ and the displayed correlator pattern. The two three-vector Gram determinants therefore give
\begin{equation}\label{eq:gram-dets-E}
\begin{gathered}
0=\det\begin{pmatrix}1&t&t\\t&1&z\\t&z&1\end{pmatrix}=(1-z)(1+z-2t^2),\\
0=\det\begin{pmatrix}1&t&-t\\t&1&z\\-t&z&1\end{pmatrix}=(1+z)(1-z-2t^2).
\end{gathered}
\end{equation}
Consequently $z=0$ and $t=r$. Hence the only behavior compatible with nonnegative weights is
\begin{equation}\label{eq:unique-E}
\vQ_{\vtheta}=(0,0,0,0,r,r,r,-r),\qquad\lambda_0=\lambda_1=\lambda_2=k.
\end{equation}
For example, it is realized by
\begin{equation}\label{eq:angles-E}
(a_0,a_1,b_0,b_1,\theta)=\left(0,\frac{\pi}{2},\frac{\pi}{4},\frac{7\pi}{4},\frac{\pi}{2}\right).
\end{equation}
The angular representation is not unique because of the common-angle shift at maximal entanglement, but the behavior and these weights are fixed. The total weight is
\begin{equation}\label{eq:weight-too-big}
\lambda_0+\lambda_1+\lambda_2=3k=\frac{10001}{10000}>1.
\end{equation}
Thus this pair admits no valid convex decomposition. This conclusion follows from solving the projected equations and testing the weights, not from evaluating a pre-existing separating Bell inequality.

\paragraph*{The first randomly selected pair.} The first pair in the same seeded ordering is $(\vL_{0000},\vL_{1101})$. This pair can also be excluded analytically. In any admissible decomposition of $\vQ_*$, its quantum component must be nonlocal: otherwise the origin and all three components are local. In particular $\lambda_0>0$ and $\theta>0$. Write
\begin{equation}\label{eq:hmn}
h=\frac{k}{\lambda_0}>0,\qquad m=\frac{\lambda_1}{\lambda_0}\ge0,\qquad n=\frac{\lambda_2}{\lambda_0}\ge0.
\end{equation}
The component has coordinates
\begin{equation}\label{eq:component-hmn}
\begin{aligned}
A_0&=n-m,&A_1&=-m-n,&B_0&=-2h-m+n,&B_1&=n-m,\\
E_{00}&=hr-m-n,&E_{01}&=h(2+r)-m-n,\\
E_{10}&=hr-m+n,&E_{11}&=-h(2+r)-m+n.
\end{aligned}
\end{equation}
At maximal entanglement, $A_1=0$ forces $m=n=0$, contradicting $B_0=0$. For $0<\theta<\pi/2$, put $z=\cos^2\theta\in(0,1)$. Eliminating the individual angle sines from Eq.~\eqref{eq:example-parameter} gives the necessary polynomial identities
\begin{equation}\label{eq:fxy-E}
f_{xy}:=z^2+(E_{xy}^2-1-A_x^2-B_y^2)z+A_x^2+B_y^2+A_x^2B_y^2-2E_{xy}A_xB_y=0.
\end{equation}
Substitution yields the exact identity
\begin{equation}\label{eq:identity-E}
(h+m)(f_{01}-f_{00})-h(f_{11}-f_{00})=-4hm(h+m)(m-n+1)^2.
\end{equation}
Since $A_0^2=(n-m)^2\le z<1$, the last squared factor cannot vanish. Thus $m=0$. This implies $A_1=-A_0$, $B_1=A_0$, and $E_{01}+E_{11}=0$. The preceding angular argument then forces $E_{00}+E_{10}=0$, contrary to $E_{00}+E_{10}=2hr>0$. Hence the first sampled pair has no admissible solution either.

\paragraph*{Scope of the calculation.} These are exact exclusions for the two specified local pairs. The accompanying multistart computation also examines all $28$ pairs in the selected CHSH facet. No admissible candidate was found, but this numerical observation is not an exhaustive proof of infeasibility. A complete primal exclusion of $\vQ_*$ requires globally certified infeasibility for every necessary local subset, including larger subsets unless a separate coverage theorem reduces the search to pairs.

The recorded single-point benchmark uses

\begin{equation}\label{eq:F-benchmark-input}
\begin{gathered}
\vtheta_*=(11\pi/6,\pi/2,5\pi/3,2\pi/5,9\pi/20),\\
\vL_{uv}=(u,v,u,v,1,uv,uv,1),\qquad u,v\in\{\pm1\},\\
\vQ_{\mathrm{test}}=0.6\vQ_{\vtheta_*}+0.2\vL_{++}+0.2\vL_{+-}.
\end{gathered}
\end{equation}
The generating angles and weights were not supplied as initial guesses. Across twenty independent seeds, all three methods returned numerical decompositions at $\epsilon_{\mathrm{num}}=10^{-9}$. The median times were $0.0128\,\mathrm{s}$ for the specified feasible pair, $0.339\,\mathrm{s}$ for a cyclic search over all $28$ pairs of the violated facet, and $0.0308\,\mathrm{s}$ for the eight-vertex search. The corresponding observed ranges were $0.0062$--$0.0293\,\mathrm{s}$, $0.0174$--$1.6491\,\mathrm{s}$, and $0.0058$--$0.0520\,\mathrm{s}$.

The following settings describe the original recorded benchmark; it was not rerun for this revision. These timings were recorded with Python~3.13.5, NumPy~2.3.5, and SciPy~1.17.0 in a virtualized Intel Xeon Platinum~8573C environment, using one process and one BLAS thread. Per-input preparation was included, while module imports were excluded. A supplementary set comprised forty nonlocal mixtures with arbitrary local distributions on all sixteen vertices, ten nonmaximally entangled members of $\cS$, and ten maximally entangled members. The initial solver found $56/60$ decompositions within budget; all four unresolved cases were in the maximally entangled group. Adding the direct branch above gave $60/60$ numerical acceptances. The source code and original result files accompany these benchmarks. They demonstrate feasibility on the tested inputs, not strict certification, a worst-case running-time guarantee, or a speed comparison with NPA.

\bibliographystyle{apsrev4-2}
\bibliography{exact_boundary_refs}
\end{document}